\documentclass[11pt]{amsart}
\usepackage{latexsym,amssymb,amsmath,youngtab,,amsthm,amscd,MnSymbol}
\usepackage{multirow}
\usepackage{youngtab,array}
\usepackage{epsfig}
\usepackage{tikz}
\usepackage{ulem} 

\usepackage[margin=2cm]{geometry}
\usepackage{tikz}
\usetikzlibrary{arrows,shapes.geometric,positioning}

\usepackage[mathscr]{eucal}
\usepackage{color}
\newcommand*{\textlabel}[2]{%
  \edef\@currentlabel{#1}
  \phantomsection
  #1\label{#2}
}
\newtheoremstyle{custom}
  {3pt}
  {3pt}
  {\slshape}
  {}
  {\bfseries}
  {.}
  { }
   {}
\theoremstyle{custom}
\newtheorem{theorem}{Theorem}[section]
\newtheorem{proposition}[theorem]{Proposition}

\newtheorem{proposition/definition}[theorem]{Proposition/Definition}

\theoremstyle{definition}

\newtheorem{example}[theorem]{Example}

\newtheorem{question}[theorem]{Question}

\theoremstyle{remark}
\newtheorem{remark}[theorem]{Remark}

\newtheoremstyle{exercise}
  {3pt}
  {6pt}
  {}
  {}
  {\bfseries}
  {:}
  { }
   {}
\theoremstyle{exercise}
\newtheorem{exercise}[theorem]{Exercise}
\newtheoremstyle{exercises}
  {3pt}
  {6pt}
  {}
  {}
  {\bfseries}
  {:}
  {\newline}
   {}

\theoremstyle{exercise}
\newtheorem{exercises}[theorem]{Exercises}

\input epsf
\def\boxit#1{\vbox{\hrule height1pt\hbox{\vrule width1pt\kern3pt
  \vbox{\kern3pt#1\kern3pt}\kern3pt\vrule width1pt}\hrule height1pt}}

\def\trank{\text{rank}}

\def\BC{\mathbb C}

\def\BP{\mathbb P}
\def\pp#1{\mathbb P^{#1}}

\def\pp#1{{\mathbb P}^{#1}}
\def\tdim{{\rm dim}}

\def\hd{,...,}

\def\inv{{}^{-1}}

\def\cG{{\mathcal G}}

\def\cS{{\mathcal S}}
\def\cP{{\mathcal P}}

\def\ZZ{\mathbb Z}

\def\11{\mathbf 1}

\def\l{\lambda}
\def\a{\alpha}

\def\o{\omega}

\def\b{\beta}
\def\g{\gamma}
\def\s{\sigma}

\def\ot{{\mathord{ \otimes } }}
\def\op{{\mathord{\,\oplus }\,}}
\def\otc{{\mathord{\otimes\cdots\otimes} }}

\def\ra{{\mathord{\;\rightarrow\;}}}

\def\La#1{\Lambda^{#1}}

\def\op{\oplus}
\def\BZ{\Bbb Z}

\def\ep{\epsilon}
\def\op{\oplus}

\def\s{\sigma}
\def\t{\tau}

\def\a{\alpha}
\def\b{\beta}

\def\g{\gamma}
\def\l{\lambda}

\def\FS{\mathfrak  S}

\def\BP{\mathbb  P}
\def\BC{\mathbb  C}

\def\pp#1{\mathbb  P^{#1}}

\def\ep{\epsilon}

\def\ci{\mathcal  I}

\def\hd{, \hdots ,}

\def\inv{{}^{-1}}

\def\La#1{\Lambda^{#1}}

\def\pp#1{\mathbb  P^{#1}}

\def\ra{\rightarrow}

\def\tperm{\operatorname{perm}}
\def\ttrace{\operatorname{trace}}

\def\tdim{\operatorname{dim}}

\def\tmod{\operatorname{mod}}

\def\trank{\operatorname{rank}}

\def\ctimes{\times \cdots\times}

\def\be{\begin{equation}}
\def\ene{\end{equation}}

\def\trank{\mathbf{R}}

\def\G{\Gamma}

\newcommand{\Id}{\operatorname{Id}}

\newcommand{\Z}{\mathbb{Z}}

\def\Mn{M_{\langle \nnn \rangle}}\def\Mtwo{M_{\langle 2 \rangle}}\def\Mthree{M_{\langle 3\rangle}}

\def\Mn{M_{\langle \nnn \rangle}}\def\Mthree{M_{\langle 3\rangle}}

\def\Mtwo{M_{\langle 2\rangle}}\def\Mthree{M_{\langle 3\rangle}}

\def\trank{{\mathrm {rank}}}

\def\nnn{\mathbf{n}}

\DeclareMathOperator*{\argmin}{arg\,min}

\newcommand{\zfzt}{$\cS_{\BZ_4\times \BZ_3}$}
\newcommand{\twofix}{$\cS_{2fix-\BZ_3}$}
\newcommand{\lad}{$\cS_{Lader-\BZ_3}$}
\newcommand{\addtlone}{Addtl.~Dec.~\#1}
\newcommand{\addtltwo}{Addtl.~Dec.~\#2}
\newcommand{\addtlthree}{Addtl.~Dec.~\#3}

\begin{document}

\author{Grey Ballard}
\address{PO Box 7311, 
Department of Computer Science, 
Wake Forest University, 
Winston-Salem, NC 27109 }
\email{ballard@wfu.edu}

\author{Christian Ikenmeyer}
\address{Max Planck Institute for Informatics, Saarland Informatics Campus, Building E1.4, D-66123 Saarbr\"ucken, Germany}
\email{cikenmey@mpi-inf.mpg.de}
 
\author{J.M. Landsberg}
\address{
Department of Mathematics\\
Texas A\&M University\\
Mailstop 3368\\
College Station, TX 77843-3368, USA}
\email{jml@math.tamu.edu}
\author{Nick Ryder}
\address{Department of Mathematics\\
UC Berkeley}
\email{nick.ryder@berkeley.edu}

\thanks{Landsberg    supported by   NSF grant  DMS-1405348. Ballard supported in part by an appointment to the Sandia National Laboratories Truman Fellowship in National Security Science and Engineering, sponsored by Sandia Corporation (a wholly owned subsidiary of Lockheed Martin Corporation) as Operator of Sandia National Laboratories under its U.S. Department of Energy Contract No. DE-AC04-94AL85000.}
\title[Geometry and matrix multiplication decompositions II]{The geometry of rank decompositions of matrix multiplication II: $3\times 3$ matrices}
\begin{abstract}
This is the second in a series of papers on rank decompositions of the matrix multiplication tensor.  
We present  new rank $23$ decompositions for the
$3\times 3$ matrix multiplication tensor $\Mthree$. All our decompositions
have symmetry groups that include the standard cyclic permutation
of factors but otherwise exhibit a range of behavior.
One of them has 11 cubes as summands and admits an unexpected symmetry group of order 12.
 
We establish basic information regarding symmetry groups of
decompositions and outline two approaches for finding
new rank decompositions of $\Mn$ for larger $\nnn$.
\end{abstract}
\maketitle

\section{Introduction}
This is the second in a planned series of papers
on the geometry of  rank decompositions of the matrix multiplication tensor
$\Mn\in \BC^{\nnn^2}\ot \BC^{\nnn^2}\ot\BC^{\nnn^2}$. Our goal is to obtain
new rank decompositions of $\Mn$ by exploiting symmetry.
For a tensor $T\in \BC^{m}\ot  \BC^{m}\ot  \BC^{m}$, the {\it rank} of $T$
is the smallest $r$ such that $T=\sum_{j=1}^r a_{ j}\ot b_{ j}\ot c_j$, with $a_j,b_j,c_j\in \BC^m$.
The rank of $\Mn$ is a standard complexity measure of matrix multiplication, in particular,
it governs the total number of arithmetic operations needed to multiply two matrices.

In this paper we present rank 23 decompositions of $\Mthree$ that have large symmetry groups,
in particular all admit the standard cyclic $\ZZ_3$-symmetry of permuting the three tensor factors.
Although many rank 23 decompositions of $\Mthree$ are known \cite{Laderman,MR837607},
none of them was known to admit the standard cyclic $\ZZ_3$-symmetry.
 
 We describe  techniques to  determine symmetry groups
 of decompositions and to determine if two decompositions
 are in the same {\it family}, as defined in \S\ref{symgpsect}.
  We also develop a framework for using representation theory to write
 down new rank decompositions for $\Mn$ for all $\nnn$.
 Similar frameworks  are  also being developed
 implicitly and explicitly in \cite{DBLP:journals/corr/Burichenko15}
 and \cite{2016arXiv161201527G}.

As discussed below, decompositions come in families. 
  DeGroote \cite{MR509013} has shown, in contrast to $\Mthree$,  the family generated by Strassen's decomposition  is
the unique family of rank seven decompositions of $\Mtwo$.  Unlike $\Mtwo$, it is still not known
if the tensor rank of $\Mthree$ is indeed $23$. The best lower bound on the rank is $16$  \cite{MR3633766}. There have been substantial 
unsuccessful efforts to find smaller decompositions
by numerical methods. 
While some researchers have taken this as evidence that $23$ might be optimal, 
it just might be the case that   rank   $22$ (or smaller) decompositions might be much
rarer than border rank $22$ decompositions, and so
using numerical search methods, and 
given initial search points, 
 when they converge,
with probability nearly one  
would converge to border rank decompositions.

In this paper we focus on decompositions with   
{\it standard cyclic symmetry}: viewed as a trilinear map, matrix multiplication is $(X,Y,Z)\mapsto \ttrace(XYZ)$, where $X,Y,Z$ are $\nnn\times\nnn$ matrices.
Since $\ttrace(XYZ)=\ttrace(YZX)$, the matrix multiplication tensor $\Mn\in 
\BC^{\nnn^2}\ot \BC^{\nnn^2}\ot\BC^{\nnn^2}=Mat_{n\times n}^{\ot 3}$ has a $\BZ_3$ symmetry by
cyclically permuting the three factors. If one applies this cyclic permutation to a tensor decomposition of $\Mn$, it is transformed to
another decomposition of $\Mn$. When a decomposition is transformed to
itself, we say the decomposition has {\it cyclic symmetry}  or    
  {\it   $\BZ_3^{std}$-invariance}.

 In \S\ref{niceexs} we present three of our examples. 
 We explain the search methods used in \S\ref{greysect}. We then   discuss
  techniques for determining symmetry groups  in \S\ref{symgpsect}
 and determine the  symmetry groups of our decompositions   in \S\ref{oursymgps}.  Our searches sometimes found equivalent decompositions
 but in different coordinates, and the same techniques enabled 
 us to identify when two decompositions are equivalent.
 Further techniques for studying decompositions are presented in \S\ref{configsect} and \S\ref{evalsect},
 respectively in terms of configurations of points in $\pp{n-1}$ and eigenvalues. 
 In \S\ref{repsect}, we precisely describe the subspace
 of $\BZ_3$-invariant tensors in $Mat_{n\times n}^{\ot 3}$, as well as the subspaces
 invariant under other finite groups. 
  In an appendix \S\ref{extraexs}, we present additional decompositions that we found.

\subsection*{Why search for decompositions with symmetry?} 
There are many examples where the optimal decompositions (or expressions)
of tensors with symmetry have some symmetry. This is true of
$\Mtwo$ \cite{DBLP:journals/corr/ChiantiniILO16,DBLP:journals/corr/Burichenko14}, the monomial $x_1\cdots x_n$
\cite{MR3440150}, also see \cite[\S 7.1]{simonsbook},
the optimal determinantal expression of $\tperm_3$ \cite{LRpermdet}, and other cases.  In any case,   imposing the symmetry
i) reduces the size of the search space, and ii) provides a guide for constructing decompositions
through the use of \lq\lq building blocks\rq\rq .

 \subsection*{Notation and conventions} $A,B,C,U,V,W$ are vector spaces, $A^*$ is the dual vector space to $A$, $GL(A)$ denotes the group of invertible
 linear maps $A\ra A$,   $SL(A)$ the maps with determinant one, and
 $PGL(A)=GL(A)/\{\BC\Id_A \backslash 0\}$ the group of projective transformations of projective space $\BP A$.
 The action of $GL(A)$ on $A\ot A^*$ descends to an action of $PGL(A)$. 
 If $a\in A$, $[a]$ denotes the corresponding point in projective space. $\FS_d$ denotes the permutation group on $d$ elements
 and $\ZZ_d  $ denotes the cyclic group of order $d$.
 For $X\subset \BP V$, $\hat X\subset V$ denotes its pre-image under the projection map union the origin.
 $X^{(\times r)}$ denotes the quotient $X^{\times r}/\FS_r$.
 We write $SL_m=SL(\BC^m)$, $GL_m=GL(\BC^m)$, and   $H^{SL_m}\subset SL_m$  is the subgroup of diagonal matrices.
 For a matrix $a$, $a^i_j$ denotes the entry in the $i$-th row and $j$-th column.

 \subsection*{Acknowledgements} Work on this paper  began  during 
the fall 2014 semester   program
 {\it Algorithms and Complexity in Algebraic Geometry} at the 
 Simons Institute for the Theory of Computing. We thank the Institute
for bringing  us together and    making this paper possible.

  \section{Examples}\label{niceexs}
  
\subsection{A rank $23$ decomposition of $\Mthree$ with $\BZ_4\times \BZ_3$ symmetry}\label{bz4bz3decomp}
Let $\BZ_4^{a_0} \subset  GL_3\subset  GL_3^{\times 3}$
be generated by
\be\label{aodef}
{a_0}=\begin{pmatrix} 0&0&-1\\ 1&0&-1\\ 0&1& -1\end{pmatrix},
\ene
where the action on   $x\ot y\ot z$ where $x,y,z$ are $3\times 3$-matrices
is
$x\ot y\ot z\mapsto {a_0} x{a_0}\inv \ot {a_0} y{a_0}\inv\ot {a_0} z{a_0}\inv$
and let $\BZ_3^{std}$ denote the standard cyclic symmetry.

Here is the decomposition, we  call it $\cS_{\BZ_4\times \BZ_3}$:
\begin{align}
\label{43a} \Mthree =&-\begin{pmatrix} 0&0&-1\\ 1&0&-1\\ 0&1& -1\end{pmatrix}^{\ot 3}\\
\label{43c} &+
\BZ_4^{a_0}/\BZ_2^{ a_0^2}\cdot \begin{pmatrix} 0&1&0\\ 0&1&0\\ 0&0&0\end{pmatrix}^{\ot 3}\\
\label{43b} &+
\BZ_4^{a_0}\cdot \begin{pmatrix} 1&0&0\\ 0&0&0\\ 0&0&0\end{pmatrix}^{\ot 3}\\
\label{43d} &+
\BZ_4^{{a_0} } \cdot \begin{pmatrix} 0&-1&0\\ 1&-1&0 \\ 0&0&0\end{pmatrix}^{\ot 3}\\
\label{43e} &+
\BZ_3^{std}\times \BZ_4^{{a_0}}\cdot \begin{pmatrix} 0&0&0\\ 0&0&1\\ 0&0&0\end{pmatrix}
\ot \begin{pmatrix} 0&0&0\\ 0&0&0\\ -1&1&0\end{pmatrix}
\ot \begin{pmatrix} 0&0&0\\ 0&1&-1\\ 0&1&-1\end{pmatrix}.
\end{align}

The decomposition is a sum of $23$ terms, each of which is a trilinear form on
matrices. For example, the   term 
$\begin{pmatrix} 0&0&0\\ 0&0&1\\ 0&0&0\end{pmatrix}
\ot \begin{pmatrix} 0&0&0\\ 0&0&0\\ -1&1&0\end{pmatrix}
\ot \begin{pmatrix} 0&0&0\\ 0&1&-1\\ 0&1&-1\end{pmatrix}$
sends a triple of matrices $(x,y,z)$
to the number $x_{23}(y_{32}-y_{31})(z_{22}-z_{23}+z_{32}-z_{33})$.
This expresses $\Mthree$ in terms of five $\BZ_4\times \BZ_3$-orbits, of sizes $1,2,4,4,12$.

\subsection{An element of the $\cS_{\BZ_4\times \BZ_3}$ family with $a_0$ diagonalized}
It is also illuminating to diagonalize $a_0$, in other words decompose the decomposition with respect to the $\BZ_4$ action:
In the following plot each row of $3 \times 3$ matrices forms an orbit under the $\BZ_4$-action. In the first four rows are 
the eleven $3\times 3$ matrices that appear as cubes in the decomposition. In remaining 4 rows each of the four columns forms three rank one tensors by tensoring the three matrices in the column in three different orders:
1-2-3, 2-3-1, and 3-1-2. 

Each complex number in each matrix is depicted by plotting its position in the complex plane.
To help identify  the precise position, a square is drawn with vertices $(-\frac 1 2,-\frac 1 2)$, $(\frac 1 2,-\frac 1 2)$, $(-\frac 1 2,\frac 1 2)$, $(\frac 1 2,\frac 1 2)$.
To quickly identify the absolute value of a complex number they are color coded:\\
\begin{tabular}{r|ccccc}
Symbol for the complex number
&
\scalebox{0.45}{\begin{tikzpicture}\node[fill,red!70,circle,minimum size=2cm] at (0,0){};\draw[black] (-1,0) -- (1,0);\draw[black] (0,-1) -- (0,1);\draw[black] (0.5,0.5) -- (0.5,-0.5) -- (-0.5,-0.5) -- (-0.5,0.5) -- cycle;\node[fill,circle,black] at (0.0,1.0) {};\end{tikzpicture}}
&
\scalebox{0.45}{\begin{tikzpicture}\node[fill,purple!40,circle,minimum size=2cm] at (0,0){};\draw[black] (-1,0) -- (1,0);\draw[black] (0,-1) -- (0,1);\draw[black] (0.5,0.5) -- (0.5,-0.5) -- (-0.5,-0.5) -- (-0.5,0.5) -- cycle;\node[fill,circle,black] at (-0.5,0.5) {};\end{tikzpicture}}
&
\scalebox{0.45}{\begin{tikzpicture}\node[fill,yellow,circle,minimum size=2cm] at (0,0){};\draw[black] (-1,0) -- (1,0);\draw[black] (0,-1) -- (0,1);\draw[black] (0.5,0.5) -- (0.5,-0.5) -- (-0.5,-0.5) -- (-0.5,0.5) -- cycle;\node[fill,circle,black] at (-0.25,0.25) {};\end{tikzpicture}}
&
\scalebox{0.45}{\begin{tikzpicture}\node[fill,blue!60,circle,minimum size=2cm] at (0,0){};\draw[black] (-1,0) -- (1,0);\draw[black] (0,-1) -- (0,1);\draw[black] (0.5,0.5) -- (0.5,-0.5) -- (-0.5,-0.5) -- (-0.5,0.5) -- cycle;\node[fill,circle,black] at (0.0,0.5) {};\end{tikzpicture}}
&
\scalebox{0.45}{\begin{tikzpicture}\node[fill,green,circle,minimum size=2cm] at (0,0){};\draw[black] (-1,0) -- (1,0);\draw[black] (0,-1) -- (0,1);\draw[black] (0.5,0.5) -- (0.5,-0.5) -- (-0.5,-0.5) -- (-0.5,0.5) -- cycle;\node[fill,circle,black] at (0.25,-0.75) {};\end{tikzpicture}}
\\
\hline
absolute value & 1 & $\frac 1 {\sqrt 2}$ & $\frac 1 {2{\sqrt 2}}$ & $\frac 1 2$ & $\frac {\sqrt 5} {2\sqrt 2}$
\end{tabular}

Numbers with a blue background are the sum of two numbers with a yellow background.
Numbers with a purple background are twice the  numbers with a yellow background.
Numbers with a red background are twice the   numbers with a blue background.
Numbers with a green background are the sum of a number with a yellow background and a number with a blue background.

Matrix cells with zeros are left empty.

\begin{minipage}{12cm}
\mbox{
\framebox{\scalebox{0.45}{%
};%
\end{tikzpicture}}}
}

\end{minipage}

\subsection{A cyclic-invariant decomposition in the Laderman family}
 Let $\cS_{Laderman}$ denote the rank $23$ Laderman decomposition of $\Mthree$. 
By \cite{DBLP:journals/corr/Burichenko14}  the symmetry group of this decomposition (see \S\ref{symgpsect}) is  $\G_{\cS_{Laderman}}=(\BZ_2\times \BZ_2)\rtimes (\BZ_3\rtimes \BZ_2)\simeq \FS_4$, 
where  $\BZ_2\times \BZ_2\subset SL_3^{\times 3}$.
We found a decomposition, which we call $\cS_{Lader-\BZ_3^{std}}$ that we identified as a $\BZ_3^{std}$-invariant
member of the Laderman {\it family}  in the sense of \S\ref{symgpsect}.

Let 
$$
\t_{23}=\begin{pmatrix} 1& 0&0\\ 0&0&1\\ 0&1&0\end{pmatrix},
\ \ep_2= \begin{pmatrix} 1& 0&0\\ 0&-1&0\\ 0&0&1\end{pmatrix}, \
 \t_{13}=\begin{pmatrix} 0& 0&1\\ 0& 1&0\\ 1&0&0\end{pmatrix}, \ 
$$
 
Let $\phi ( x\ot y\ot z)=\t_{13}x\t_{13}\ot \t_{13}y\ot z\t_{13}$, and
let $\zeta(x\ot y\ot z)= \ep_2y^T \ep_2\ot \ep_2 x^T\ep_2\ot \ep_2z^T \ep_2$. Let $\pi$ be the
generator of $\BZ_3^{std}$. Let $\G$ denote  the group generated by $\pi,\phi$, and $\zeta$.

\begin{align}
\label{nlada} \Mthree&=\begin{pmatrix} 0 &0&0\\0&1&0\\ 0&0&0\end{pmatrix}^{\ot 3}
\\
\label{nladb} &+
 \G/(\BZ_3^{\pi}\rtimes \BZ_2^{\zeta})\cdot \begin{pmatrix} 0 &0&0\\0&0&0\\ 0&0&1\end{pmatrix}^{\ot 3}
\\
\label{nladc} & +\G/(\BZ_3^{\pi}\rtimes \BZ_2^{\zeta}) \cdot \begin{pmatrix} -1 & 1&0\\-1&0&0\\ 0&0&0\end{pmatrix}^{\ot 3}
\\
\label{nladd} &+
 \G /(\BZ_2^{\phi}\rtimes \BZ_2^{\zeta})\cdot  \begin{pmatrix} 0 &-1&0\\0&0&0\\ 0&0&0\end{pmatrix}
\ot \begin{pmatrix} 1 &-1&0\\ 1&-1&-1\\ 0&1&1\end{pmatrix}
\ot \begin{pmatrix} 0 &0&0\\1&0&0\\ 0&0&0\end{pmatrix}
\\
\label{nlade} &+
\G/ \BZ_3^{\pi} \cdot \begin{pmatrix} 1 &0&0\\1&0&0\\ 0&0&0\end{pmatrix}^{\ot 3}
\end{align}
These five  orbits are respectively of sizes $1,4,4,6,8$. There are five $\BZ_3^{std}$ invariant terms,
one from each of \eqref{nlada},\eqref{nladb},\eqref{nladc} and two from \eqref{nlade} because $\zeta$ preserves
$\BZ_3^{std}$-invariance. 

We originally found this  decomposition   by numerical methods. The incidence graphs
discussed in \S\ref{graphsect} gave us
evidence that it should be in the Laderman family, and then it was straightforward to find the transformation that
exchanged $\cS_{Laderman}$ and $\cS_{Lader-\BZ_3^{std}}$,   namely 
$x\ot y\ot z\mapsto \t_{12}x\ot y \ep_2\t_{12}\ot\t_{12}\ep_2 z\t_{12}$.
See \S\ref{graphsect} for more discussion.
As shown in in \cite{DBLP:journals/corr/Burichenko14}, $\G$  is isomorphic
to $\FS_4$ and these are all the symmetries of the Laderman family. 
Translated to $\cS_{Laderman}$ as presented in 
   \cite{DBLP:journals/corr/Burichenko14}, these five orbits are respectively $\{ 19\}$,
$\{20,21,22,23\}$, $\{ 4,7,12,16\}$, $\{ 1,3,6,10,11,14\}$ and 
$\{ 2,5,8,9,13,15,17,18\}$. 

\begin{remark} 
As pointed out in \cite{2017arXiv170308298S}, there are similarities between this decomposition and Strassen's. 
\end{remark}

\begin{remark}  
The   decompositions of Johnson-McLouglin \cite{MR837607} cannot have $\BZ_3$-invariant decompositions   for rank reasons:
The first space of decompositions  has five terms  (those numbered 3,12,16,22,23 in \cite{MR837607})  
where there are two matrices of rank one and one of rank greater than one, while to have any external
$\BZ_3$-symmetry, the number of such would have to be a multiple of three.
The second space has a unique matrix of rank three appearing (23b), so is ruled out for the same reason.
\end{remark}

\subsection{Decomposition  $\cS_{2fix-\BZ_3}$}
 Here is a decomposition with two $\BZ_3$-fixed points:

\begin{align}
\Mthree
&= \label{twofix_c1} \begin{pmatrix} 1 & 0 & 0 \\ 0 & 0 & 1 \\ 0 & 0 & 1 \end{pmatrix}^{\ot 3} \\
&+ \label{twofix_c2} \begin{pmatrix} 0 & 0 & 0 \\ 0 & 1 & -1 \\ 0 & 0 & 0 \end{pmatrix}^{\ot 3} \\
&+ \label{twofix_c3} \BZ_3^{std} \cdot
	\begin{pmatrix} 0 & 1 & 0 \\ 0 & 0 & 1 \\ 0 & 0 & 0 \end{pmatrix} \ot 
	\begin{pmatrix} 0 & 0 & 0 \\ 0 & 1 & -1 \\ 0 & 1 & -1 \end{pmatrix}
    \ot \begin{pmatrix} 0 & 0 & 0 \\ 1 & 0 & -1 \\ 0 & 0 & 0 \end{pmatrix}  \\ 
&+ \label{twofix_c5} \BZ_3^{std} \cdot 
	\begin{pmatrix} 0 & -1 & 1 \\ 0 & 0 & 0 \\ 0 & 0 & 0 \end{pmatrix} \ot 
	\begin{pmatrix} 0 & 0 & 0 \\ 0 & 0 & 0 \\ 0 & 0 & 1 \end{pmatrix} \ot 
	\begin{pmatrix} 0 & 0 & 0 \\ 0 & 0 & 0 \\ 1 & 0 & 0 \end{pmatrix} \\     
&+ \label{twofix_c6} \BZ_3^{std} \cdot 
	\begin{pmatrix} 1 & 0 & 0 \\ 1 & 0 & 0 \\ 0 & 0 & 0 \end{pmatrix} \ot 
	\begin{pmatrix} 0 & -1 & 1 \\ 0 & 0 & 0 \\ 0 & 0 & 0 \end{pmatrix} \ot 
	\begin{pmatrix} 0 & 0 & 0 \\ 0 & 0 & 0 \\ 0 & 1 & 0 \end{pmatrix} \\ 
&+ \label{twofix_c4} \BZ_3^{std} \cdot 
	\begin{pmatrix} 1 & 0 & 0 \\ 0 & 0 & 1 \\ 0 & 0 & 0 \end{pmatrix} \ot 
	\begin{pmatrix} 0 & 1 & 0 \\ 0 & 0 & 1 \\ 0 & 0 & 1 \end{pmatrix} \ot 
	\begin{pmatrix} 0 & 0 & 0 \\ 1 & 0 & -1 \\ 0 & 1 & -1 \end{pmatrix} \\ 
&+ \label{twofix_c7} \BZ_3^{std} \cdot 
	\begin{pmatrix} 1 & 0 & 0 \\ 0 & 0 & 0 \\ 0 & 0 & 0 \end{pmatrix} \ot 
	\begin{pmatrix} 0 & 0 & 1 \\ 0 & 0 & 1 \\ 0 & 0 & 1 \end{pmatrix} \ot 
	\begin{pmatrix} 0 & 0 & 0 \\ 0 & 0 & 0 \\ 1 & -1 & 0 \end{pmatrix} \\ 
&+ \label{twofix_c8} \BZ_3^{std} \cdot 
	\begin{pmatrix} 0 & 0 & 0 \\ 0 & 0 & 1 \\ 0 & 0 & 0 \end{pmatrix} \ot 
	\begin{pmatrix} 0 & 1 & 0 \\ 0 & 1 & 0 \\ 0 & 1 & 0 \end{pmatrix} \ot 
	\begin{pmatrix} 0 & 0 & 0 \\ -1 & 1 & 0 \\ 0 & 0 & 0 \end{pmatrix} \\ 
&+ \label{twofix_c9} \BZ_3^{std} \cdot 
	\begin{pmatrix} 0 & 0 & 0 \\ 0 & 0 & 1 \\ 0 & 0 & 1 \end{pmatrix} \ot 
	\begin{pmatrix} 1 & 0 & 0 \\ 1 & 0 & 0 \\ 1 & 0 & 0 \end{pmatrix} \ot 
	\begin{pmatrix} -1 & 1 & 0 \\ 0 & 0 & 0 \\ 0 & 0 & 0 \end{pmatrix} 
\end{align}

An interesting feature of this decomposition is that it \lq\lq nearly\rq\rq\  has a transpose-like symmetry, as discussed in later sections.

\section{Discussion on the numerical methods used}
\label{greysect}

Our techniques for discovering cyclic-invariant decompositions use numerical optimization methods that are designed to compute approximations rather than exact decompositions.
We also use heuristics in order to encourage sparsity in the solutions, and a fortunate by-product of the sparsity is that the nonzero values often tend towards a discrete set of values from which an exact decomposition can be recognized.
Our methods are based on techniques that have proved successful in discovering generic exact decompositions (those that have no noticeable symmetries) \cite{BB15,Smirnov13}; we summarize this approach in Section \ref{sec:ALS}.
Our search process can be divided into two phases.
First, as we discuss in Section \ref{sec:LOQO}, we find a dense, cyclic-invariant, approximate solution using nonlinear optimization methods.
Then, as we describe in Section \ref{sec:sparsify}, we transform the dense approximate solution to an exact cyclic-invariant decomposition using heuristics that encourage sparsity.

\subsection{Alternating Least Squares}
\label{sec:ALS}

Computing low-rank approximations of tensors is a common practice in data analysis when the data represents multi-way relationships.
The most popular algorithm for computing approximations with CANDECOMP/PARAFAC (CP) structure, i.e., rank decompositions,  is known as alternating least squares (ALS) \cite{KB09}.
In particular, in the case of the matrix multiplication tensor, the objective function of the optimization problem is given by
\begin{equation}
\label{eq:objfun}
\argmin_{X,Y,Z} \left\|\Mn-\sum_{r=1}^R x_r\ot y_r\ot z_r\right\|,
\end{equation}
where $X$, $Y$, and $Z$ are $\nnn^2\times R$ \emph{factor matrices} with $r$th columns given by $x_r$, $y_r$, and $z_r$ and
this and all  norms correspond  to the square root of the sum of squares of the entries of the tensor (or matrix).
The objective function itself is nonlinear and non-convex and cannot be solved in closed form.
However, if two of the three factor matrices are held fixed, then the resulting objective function is a linear least squares problem and can be solved using linear algebra.
Thus, ALS works by alternating over the factor matrices, holding two fixed and updating the third, and iterating until convergence.

In general, ALS iterations are performed in floating point arithmetic.
While an objective function value of 0 corresponds to a rank decomposition, with ALS we can hope for an objective function value only as small as the finite precision allows.
In the case of the matrix multiplication tensor, there are multiple pitfalls that make it difficult to find approximations that approach objective function values of 0.
The most successful technique was proposed by Smirnov and uses regularization, with objective function given by
 \begin{equation}
 \label{eq:objfunreg}
 \argmin_{X,Y,Z} \left\|\Mn-\sum_{r=1}^R x_r\ot y_r\ot z_r\right\| + \lambda \left( \|X-\tilde X\|  + \|Y-\tilde Y\|  + \|Z - \tilde Z\|  \right),
 \end{equation}
 for judicious choices of scalar $\lambda$ and matrices $\tilde X$, $\tilde Y$, and $\tilde Z$ \cite{Smirnov13}.
 We discuss effective choices for the regularization parameters in Section \ref{sec:sparsify}.
 This method works for the cases of non-square matrix multiplication and has been used to discover exact rank decompositions for many small cases \cite{BB15,Smirnov13}, but it does not encourage solutions to reflect any symmetries.

\subsection{Nonlinear Optimization for Dense Approximations}
\label{sec:LOQO}

To enforce cyclic invariance on rank decompositions, we can impose structure on the factor matrices $X$, $Y$, and $Z$.
In particular, if a cyclic-invariant approximation includes the component $x_r\ot y_r\ot z_r$, then it must also include $z_r\ot x_r\ot y_r$ and $y_r\ot z_r\ot x_r$.
This implies either that all three components appear or that $x_r=y_r=z_r$, which means the factor matrices have the following structure:
\begin{align}
X &= \begin{pmatrix} A & B & C & D \end{pmatrix} \notag \\
Y &= \begin{pmatrix} A & D & B & C \end{pmatrix} \label{eq:cistruct} \\
Z &= \begin{pmatrix} A & C & D & B \end{pmatrix}, \notag
\end{align}
where $A$ is an $\nnn^2\times P$ matrix and $B,C,D$ are $\nnn^2\times Q$ matrices with $P+3Q=R$.
With this structure, \eqref{eq:objfun} becomes
\begin{equation}
\label{eq:objfunci}
\argmin_{A,B,C,D} \left\|\Mn-\sum_{p=1}^P a_p\ot a_p\ot a_p - \sum_{q=1}^Q \left(b_q\ot c_q\ot d_q + d_q\ot b_q\ot c_q + c_q\ot d_q\ot b_q\right)\right\|.
\end{equation}
Note that the total number of variables (now spread across 4 matrices) is reduced by a factor of 3.

Again, this objective function is nonlinear and non-convex.
It is possible to use an ALS approach to drive the objective function value to 0; however, while \eqref{eq:objfunci} is linear in $B$, $C$, and $D$, it is not linear in $A$.
Thus, the optimal update for $A$, for fixed $B$, $C$, and $D$, cannot be computed in closed form.
While there are numerical optimization techniques for $A$, we were not successful in using an ALS approach to drive the objective function value close to zero.

Instead, we used generic nonlinear optimization software to search for cyclic-invariant approximations.
In particular, we used the LOQO software package \cite{Vanderbei99}, which relies on the AMPL \cite{FGK02} modeling language to specify the optimization problem.
In order to find solutions, we try all possible values of $P$ and $Q$, use multiple random starting points, and constrain the variable entries to be no greater than 1 in absolute value.
For fixed $R$, there are $\lfloor R/3 \rfloor$ possible values for $P$ and $Q$.
Multiple starting points are required because the objective function is non-convex: numerical optimization techniques are sensitive to starting points in this case, with approximations often getting stuck at local minima.
Constraining the variables to be no greater than 1 in absolute value is a technique to avoid converging numerically to border rank decompositions, in which case some variable values must grow in magnitude to continue to improve the objective function value.
Driving the objective function value to zero with bounded variable values ensures that the approximation corresponds to a rank decomposition.

When we are successful, we obtain matrices $A,B,C,D$ that correspond to an objective function value very close to zero.
Thus, the approximation has the cyclic invariance we desire.
However, these matrices are dense and have floating point values throughout the $[{-}1,1]$ range.
Rounding these floating point values, even to a large set of discrete rational values, typically does not yield an exact rank decomposition.
The next section describes the techniques we use to convert dense approximate solutions to exact solutions.

\subsection{Heuristics to Encourage Sparsity for Exact Decompositions}
\label{sec:sparsify}

In order to obtain exact decompositions, we use the ALS regularization heuristics that proved effective for the non-invariant case, given in equation \eqref{eq:objfunreg}.
However, those techniques ignore the cyclic-invariant structure in the dense approximations we obtain from the techniques described in \S\ref{sec:LOQO}.
The heuristic we used to discover cyclic-invariant rank decomposition consists of alternating between ALS iterations with regularization and projection of the approximation back to the set of cyclic-invariant solutions.

We now describe the specifics of the regularization terms from \eqref{eq:objfunreg}.
The scalar parameter $\lambda$ determines the relative importance between an accurate approximation of the matrix multiplication tensor and adherence to the regularization terms.
In the context of ALS, only one of the regularization terms affects the optimization problem when updating one factor matrix.
The \emph{target matrices} $\tilde X$, $\tilde Y$, $\tilde Z$ are designed to encourage the corresponding factor matrices to match a desired structure, and they can be defined differently for each iteration of ALS.
Here is the technique proposed by Smirnov \cite{Smirnov13}.
Consider the update of factor matrix $X$.
By default, $\tilde X$ is set to have the values of $X$ from the previous iteration.
If any of the values are larger than 1 (or any specified maximum value), then the corresponding value of $\tilde X$ is set to magnitude 1 with the corresponding sign.
Then, for a given number $z$ of desired zeros, the smallest $z$ values of $\tilde X$ are set to exactly 0.
Thus, the regularization term $\|X-\tilde X\|_F$ will encourage any large values of $X$ to tend towards $\pm 1$ and the smallest $z$ values to tend toward 0.
The parameter $z$ can be varied over iterations and also across factor matrices (though in the case of cyclic-invariant approximations the set of values within each factor matrix is the same as those of the other factor matrices).

Because ALS does not enforce cyclic invariance, the approximation will tend to deviate from the structure given in equation \eqref{eq:cistruct}.
We project back to a cyclic-invariant approximation by setting all three values that should be the same in each of the factor matrices to their average value.

Fortunately, perhaps miraculously, by encouraging sparsity and bounding the variable values, when many entries are driven to zero, the nonzero values tend towards a discrete set of values (usually $\pm1$).
However, the process of converting a dense approximation to an exact decomposition involves manual tinkering and much experimentation.
The basic approach we have used successfully is to maintain a tight approximation to the matrix multiplication tensor, start with $z=0$, gradually increase $z$, frequently project back to cyclic invariance, and play with the $\lambda$ parameter between values of $10^{-3}$ and $1$.
ALS iterations are relatively cheap, so often 100 or 1000 iterations can be taken with a given parameter setting before changing the configuration.
While this process is artful and lacks any guarantees of success, we have nearly always been able to convert cyclic-invariant dense approximations of $\Mthree$ to exact decompositions.

\section{Symmetry groups of tensors and decompositions}\label{symgpsect}
  
In this section we explain how we found the additional symmetries beyond the $\BZ_3^{std}$ that was built into the search, 
and describe the full  
symmetry groups of the decompositions. We establish additional properties regarding symmetries
for use in future work.
We begin with a general discussion of symmetry groups of tensors and their decompositions.

\subsection{Symmetry groups of tensors}
Let $V$ be a complex vector space and let  $T\in V^{\ot k}$.  Define the {\it symmetry group} of $T$,  $G_T\subset GL(V)^{\times k} \rtimes \FS_k$ to be the subgroup preserving $T$,
where $\FS_k$ acts by permuting the factors.

 For a rank decomposition   $T=\sum_{j=1}^r t_j$, where each $t_j$ has rank one, i.e., $t_j=v_{1j}\otc v_{kj}$, define the  set 
$\cS :=\{t_1\hd t_r\}$, which we also call the decomposition,  and
the {\it symmetry group of the decomposition}  $\G_{\cS }:=\{ g\in G_T\mid g\cdot \cS =\cS \}$. 
We also consider $\cS$ as a point of  the variety
of $r$-tuples of unordered points on the Segre variety of rank one tensors, denoted $\hat Seg(\BP V\ctimes \BP V)^{(\times r)}$.

If $g\in G_T$, then  $g\cdot \cS :=\{ gt_1\hd gt_r\}$ is also a rank decomposition of $T$,
and  $\G_{g\cdot \cS }= g\G_{\cS } g\inv$ (see \cite{DBLP:journals/corr/ChiantiniILO16}), so 
decompositions come in  {\it  families}
parametrized by $G_T/\G_{\cS}$, and each member of the family has the same abstract symmetry group.  
We reserve the term {\it family} for $G_T$-orbits. The  quasi-projective 
subvariety $\Sigma_r^T\subset \hat Seg(\BP A_1\ctimes \BP A_k)^{(\times r)}$ of all rank $r$ decompositions is  a union of $G_T$ orbits. 

If one is not concerned with the rank of a decomposition, then for any finite subgroup $\G\subset G_T$,
$T$ admits rank decompositions $\cS$ with $\G\subseteq \G_{\cS}$, by taking any rank decomposition of $T$ and then averaging
it with its $\G$-translates. 
We will be concerned with rank decompositions that are minimal or close to minimal, so only very special groups $\G$ can occur.
 
 \subsection{Matrix multiplication}\label{mmultsym}
   The symmetry groups of matrix multiplication decompositions are useful for determining if two decompositions
lie in the same family, and the   groups 
that appear in   known  decompositions will be a guide for constructing decompositions
in future work.

    The symmetry groups of many of the decompositions
that have already appeared in the literature are determined in \cite{DBLP:journals/corr/Burichenko15}.   Burichenko does not use the associated
graphs discussed below. One can recover the results
of \cite{DBLP:journals/corr/Burichenko15} with shorter proofs by using them.

Let $T=\Mn \in \BC^{\nnn^2}\ot \BC^{\nnn^2}\ot \BC^{\nnn^2}=:A\ot B\ot C=(U^*\ot V)\ot (V^*\ot W)\ot (W^*\ot U)$ where $U,V,W=\BC^\nnn$.
Recall that $\Mn$ is the re-ordering of $\Id_U\ot \Id_V\ot \Id_W$ and
\be\label{gmn} G_{\Mn}=[(PGL(U)\times PGL(V)\times PGL(W))\rtimes \BZ_3]\rtimes \BZ_2\subset  GL(A)\times GL(B)\times GL(C) \rtimes \FS_3,
\ene
see, e.g.,  \cite[Thms. 3.3,3.4]{MR0506377}, \cite[Prop. 4.1]{DBLP:journals/corr/ChiantiniILO16}, \cite[Thm. 2.9]{MR3513064} or 
\cite[Prop. 4.7]{DBLP:journals/corr/Burichenko15}. 
   The $\BZ_2\subset  GL(A)\times  GL(B)\times  GL(C)\rtimes \FS_3$  
may be generated   by  e.g., $(a\ot b\ot c)\mapsto (a^T\ot c^T\ot b^T)$, and the $\BZ_3$   by  cyclically permuting the factors
$A,B,C$. 
The $\BZ_3\rtimes \BZ_2$ is isomorphic to $\FS_3$, but it is more naturally thought of 
as $\BZ_3\rtimes \BZ_2$, since it is
not the $\FS_3$ appearing in \eqref{gmn}.

Thus if  $\cS $ is    a  rank decomposition of $\Mn$, then
$\G_{\cS}\subset [(PGL(U)\times PGL(V)\times PGL(W))\rtimes \BZ_3]\rtimes \BZ_2$.

\begin{remark}
As pointed out by  Burichenko \cite{DBLP:journals/corr/Burichenko14}, for matrix multiplication 
rank decompositions one can define  an {\it extended 
symmetry group} $\hat\G_{\cS}$ by viewing $A\ot B\ot C=
 (U^*\ot U)^{\ot 3}= (U^*)^{\ot 3}\ot (U^{\ot 3})$.
 We do not study such groups in this paper. 
\end{remark}

 We call a $\BZ_2\subset \G_{\cS}$ a {\it   transpose like symmetry} if
it corresponds to  the symmetry
of $\Mn$ given by $ x\ot y\ot z \mapsto  x^T\ot z^T\ot y^T $, or a cyclic variant of it such as $x\ot y\ot z\mapsto
y^T\ot x^T\ot z^T$,   composed with an element of $PGL(U)\times PGL(V)\times PGL(W)$ such that the total map is an involution
on the elements of $\cS$.
Transpose  like symmetries where the elements of $PGL(U)\times PGL(V)\times PGL(W)$
are all the identity (which we will call {\it convenient transpose symmetries})  do not appear to be compatible with standard cyclic symmetries in
minimal decompositions,
at least this is the case for rank seven decompositions
of $\Mtwo$ and the known rank $23$  decompositions of $\Mthree$.

\begin{example} Let $x^i_j$, $1 \leq i,j \leq n$ be a basis of $A$,
$y^i_j$ a basis of $B$, and $z^i_j$ a basis of $C$.
Consider the standard decomposition $\cS_{std}$ of $\Mn$ of size $\nnn^3$:
\begin{equation}\label{eq:standardmamu}
\Mn=\sum_{i,j,k=1}^\nnn x^i_j\ot y^j_k\ot z^k_i. 
\end{equation}
Let $H^{SL_n}\subset SL_n$ denote the maximal torus (diagonal matrices).
It is clear
$\G_{\cS_{std}}\supseteq (H^{SL_n}\rtimes \FS_n)^{\times 3}\rtimes(\BZ_3\rtimes \BZ_2)$
because for $((\l_1\hd \l_n), (\mu_1\hd \mu_n), (\nu_1\hd \nu_n))\in (H^{SL_n} )^{\times 3}$, we have
$$(\l_i\mu_j\inv x^i_j)\ot (\mu_j\nu_k\inv y^j_k)\ot (\nu_k\l_i\inv z^k_i)=x^i_j\ot y^j_k\ot z^k_i$$
and for all $\s,\t,\eta\in \FS_n$, we have the equality of sets
$\{\cup_{i,j,k}   x^i_j\ot y^j_k\ot z^k_i\}= \{\cup_{i,j,k}   x^{\s(i)}_{\t(j)}\ot y^{\t(j)}_{\eta(k)}\ot z^{\eta(k)}_{\s(i)}\}$,   the cyclic symmetry is evident, and the $\BZ_2$ may be
generated e.g., by $x\ot y\ot z\mapsto x^T\ot z^T\ot y^T$.
\end{example}

The   {\it Comon conjecture} \cite{BCMT,MR3092255}, in its
original form,   asserts that a tensor $T\in (\BC^{N})^{\ot d}$ that happens to be symmetric, will 
have an optimal rank decomposition consisting of rank one symmetric tensors, that is, the symmetric tensor
rank of $T$ equals the usual tensor rank.
An explicit counter-example to this  when $N=800$  has been
asserted in  \cite{2017arXiv170508740S}. Nevertheless, there appear to be many instances where it
is known to hold, see, e.g.,  
\cite{MR3474849,MR3092255}
so we pose the following question:

\begin{question}[Generalized Comon Question] 
Given $T\in (\BC^N)^{\ot d}$ that is invariant under some
$\G\subset \FS_d$,  when does there exist  an optimal rank decomposition $\cS$ of $T$ that is
$\G$-invariant, i.e.,
$\G\subseteq \G_{\cS}$? \end{question}

\subsection{Invariants associated to a decomposition of $\Mn$}\label{graphsect}
Let $\Mn=\sum_{j=1}^r t_j$ be a rank decomposition $\cS$  for $\Mn$ and write $t_j=a_j\ot b_j\ot c_j$.
 Partition $\cS$ by   rank triples into
disjoint subsets: $\cS:=\{ \cS_{1,1,1}, \cS_{1,1,2}\hd \cS_{n,n,n}\}$, where  for $s \leq t \leq u$ we set  $\cS_{s,t,u}=\{ t_j\mid 
\{ \trank(a_j),\trank(b_j),\trank(c_j)\}=\{s,t,u\}  \}$.
Then 
$\G_{\cS}$ preserves each $\cS_{s,t,u}$.

We can say more about $\G_{\cS_{1,1,1}}$:  
If $a\in U^*\ot V$ and  $\trank(a)=1$, then there are unique points $[\mu]\in \BP U^*$ and
$[v]\in \BP V$ such that $[a]=[\mu\ot v]$, so 
define $\cS_{U^*}\subset \BP U^*$ and $\cS_U\subset \BP U$ to correspond to the
elements in $\BP U^*$ (resp. $\BP U$) appearing in $\cS_{1,1,1}$. 
Let $\tilde\cS_{U^*}\subset \BP U^*$ and $\tilde \cS_U\subset \BP U$ correspond
to the elements appearing in some rank one matrix in $\cS$.

Since we are concerned with decompositions with a standard cyclic symmetry, we will have
$\cS_U\simeq \cS_V\simeq \cS_W$ and $\cS_{U^*}\simeq \cS_{V^*}\simeq \cS_{W^*}$ and similarly
for the tilded spaces.

\medskip

Define a bipartite graph  $\ci\cG_{\cS}$,
 the {\it incidence graph}
where the top vertex set is given by elements in $\tilde \cS_{U^*}$ (or
  $\cS_{U^*}$) and the bottom vertex set by elements in 
$\tilde \cS_U$. Draw an edge between elements $[\mu]$ and $[v]$
if they are {\it incident}, i.e.,  $\mu(v)=0$. Geometrically, $[v]$ belongs
to the hyperplane determined by $[\mu]$ (and vice-versa). One  can weight the vertices
of this graph in several ways, the simplest 
is just by the number of times the element appears in the decomposition.
In practice (see the examples below) this has been enough to determine the
symmetry group $\G_{\cS}$, in the sense that it cuts the possible size of the group down
and it becomes straightforward to determine $ \G_\cS$ as a subgroup of the  symmetry group of
$\ci\cG_{\cS}$.

Incidence graphs for the four decompositions are given in Figures \ref{fig:std_inc}, \ref{fig:zfzt_inc}, \ref{fig:lad_inc}, and \ref{fig:twofix_inc}.
Here interpret the top set of vectors  as column vectors and the bottom set  as row vectors. Then $\mu$ from the top set  is incident to $v$ 
from the bottom set if the scalar $v\mu$ is zero, and then $\mu$ and $v$ are joined by an edge in the graph.

\tikzset{main node/.style = {rectangle, draw, text width=4em, text centered, rounded corners, minimum height=4em, minimum height=2em}}
\tikzset{line/.style = {draw,very thick}}


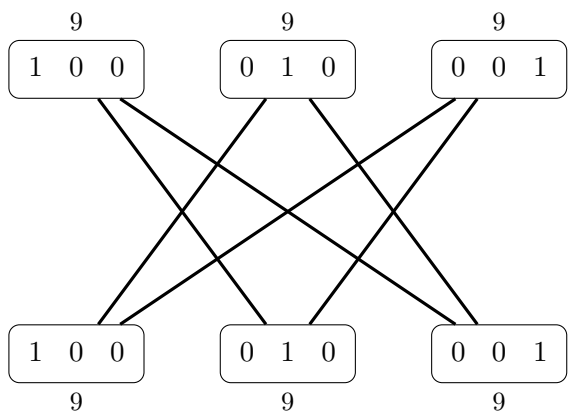
\begin{figure}
\begin{center}
\begin{tikzpicture}[node distance=3cm and 1cm, auto]

	\node[main node, label={\small 9}] (1u) {$\begin{matrix} 1 & 0 & 0 \end{matrix}$};
	\node[main node, label={\small 9}] (2u) [right=of 1u] {$\begin{matrix} 0 & 1 & 0 \end{matrix}$};
	\node[main node, label={\small 9}] (3u) [right=of 2u] {$\begin{matrix} 0 & 0 & 1 \end{matrix}$};

	\node[main node, label=below:{\small 9}] (1v) [below=of 1u] {$\begin{matrix} 1 & 0 & 0 \end{matrix}$};
	\node[main node, label=below:{\small 9}] (2v) [below=of 2u] {$\begin{matrix} 0 & 1 & 0 \end{matrix}$};
	\node[main node, label=below:{\small 9}] (3v) [below=of 3u] {$\begin{matrix} 0 & 0 & 1 \end{matrix}$};

	\path[line]
	(1u)
		edge node {} (2v)
		edge node {} (3v)
	(2u)
		edge node {} (1v)
		edge node {} (3v)
	(3u)
		edge node {} (1v)
		edge node {} (2v)
	;

\end{tikzpicture}
\end{center}
\caption{Incidence graph of standard decomposition}
\label{fig:std_inc}
\end{figure}

\begin{figure}
\begin{center}
\begin{tikzpicture}[node distance=3cm and 1cm, auto]

	\node[main node, label={\small 3}] (1u) {$\begin{matrix} 1 & 0 & 0 \end{matrix}$};
	\node[main node, label={\small 3}] (2u) [right=of 1u] {$\begin{matrix} 0 & 1 & 0 \end{matrix}$};
	\node[main node, label={\small 3}] (3u) [right=of 2u] {$\begin{matrix} 0 & 0 & 1 \end{matrix}$};
	\node[main node, label={\small 3}] (4u) [right=of 3u] {$\begin{matrix} 1 & 1 & 0 \end{matrix}$};
	\node[main node, label={\small 3}] (5u) [right=of 4u] {$\begin{matrix} 0 & 1 & 1 \end{matrix}$};
	\node[main node, label={\small 3}] (6u) [right=of 5u] {$\begin{matrix} 1 & 1 & 1 \end{matrix}$};

	\node[main node, label=below:{\small 4}] (1v) [below=of 1u] {$\begin{matrix} 1 & 0 & 0 \end{matrix}$};
	\node[main node, label=below:{\small 4}] (2v) [below=of 2u] {$\begin{matrix} 0 & 0 & 1 \end{matrix}$};
	\node[main node, label=below:{\small 4}] (3v) [below=of 3u] {$\begin{matrix} 1 & -1 & 0 \end{matrix}$};
	\node[main node, label=below:{\small 4}] (4v) [below=of 4u] {$\begin{matrix} 0 & 1 & -1 \end{matrix}$};
	\node[main node, label=below:{\small 1}] (5v) [below=of 5u] {$\begin{matrix} 0 & 1 & 0 \end{matrix}$};
	\node[main node, label=below:{\small 1}] (6v) [below=of 6u] {$\begin{matrix} 1 & 0 & -1 \end{matrix}$};

	\path[line]
	(1u)
		edge node {} (2v)
		edge node {} (4v)
		edge node {} (5v)
	(2u)
		edge node {} (1v)
		edge node {} (2v)
		edge node {} (6v)
	(3u)
		edge node {} (1v)
		edge node {} (3v)
		edge node {} (5v)
	(4u)
		edge node {} (2v)
		edge node {} (3v)
	(5u)
		edge node {} (1v)
		edge node {} (4v)
	(6u)
		edge node {} (3v)
		edge node {} (4v)
		edge node {} (6v)
	;

\end{tikzpicture}
\end{center}
\caption{Incidence graph of decomposition $\cS_{\BZ_4\times \BZ_3}$}
\label{fig:zfzt_inc}
\end{figure}
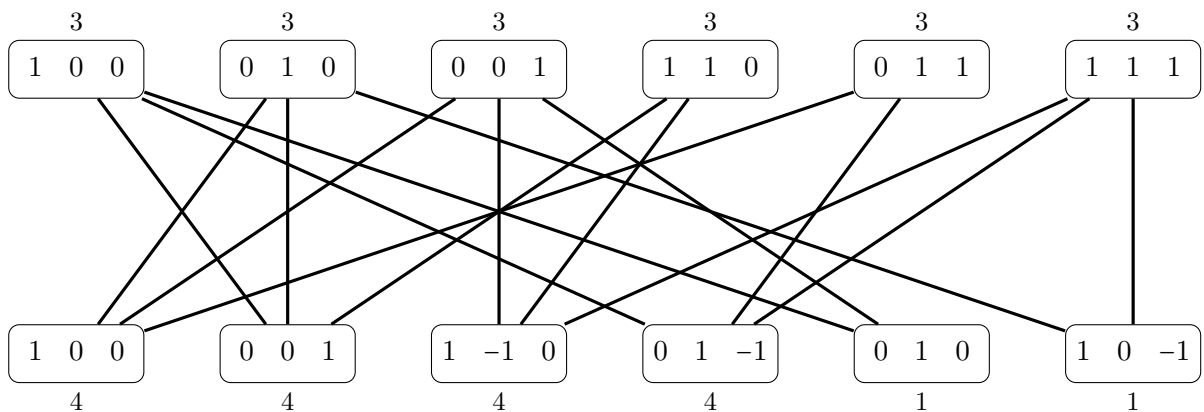

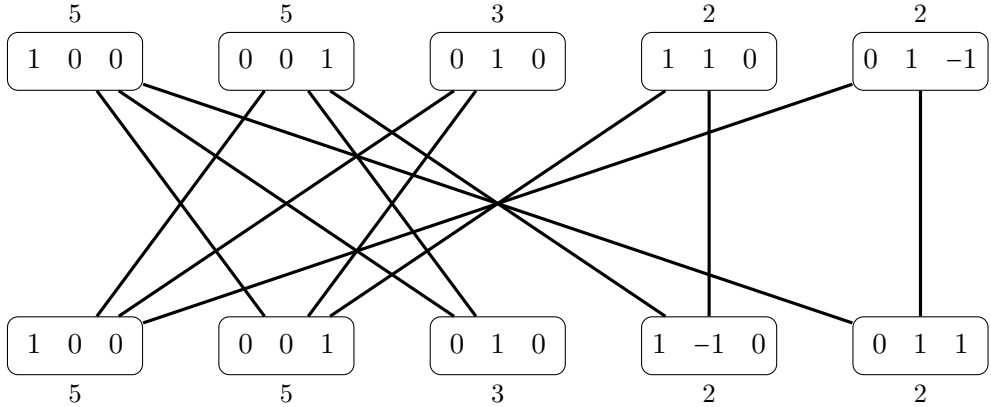
\begin{figure}
\begin{center}
\begin{tikzpicture}[node distance=3cm and 1cm, auto]

	\node[main node, label={\small 5}] (1u) {$\begin{matrix} 1 & 0 & 0 \end{matrix}$};
	\node[main node, label={\small 5}] (2u) [right=of 1u] {$\begin{matrix} 0 & 0 & 1 \end{matrix}$};
	\node[main node, label={\small 3}] (3u) [right=of 2u] {$\begin{matrix} 0 & 1 & 0 \end{matrix}$};
	\node[main node, label={\small 2}] (4u) [right=of 3u] {$\begin{matrix} 1 & 1 & 0 \end{matrix}$};
	\node[main node, label={\small 2}] (5u) [right=of 4u] {$\begin{matrix} 0 & 1 & -1 \end{matrix}$};

	\node[main node, label=below:{\small 5}] (1v) [below=of 1u] {$\begin{matrix} 1 & 0 & 0 \end{matrix}$};
	\node[main node, label=below:{\small 5}] (2v) [below=of 2u] {$\begin{matrix} 0 & 0 & 1 \end{matrix}$};
	\node[main node, label=below:{\small 3}] (3v) [below=of 3u] {$\begin{matrix} 0 & 1 & 0 \end{matrix}$};
	\node[main node, label=below:{\small 2}] (4v) [below=of 4u] {$\begin{matrix} 1 & -1 & 0 \end{matrix}$};
	\node[main node, label=below:{\small 2}] (5v) [below=of 5u] {$\begin{matrix} 0 & 1 & 1 \end{matrix}$};

	\path[line]
	(1u)
		edge node {} (2v)
		edge node {} (3v)
		edge node {} (5v)
	(2u)
		edge node {} (1v)
		edge node {} (3v)
		edge node {} (4v)
	(3u)
		edge node {} (1v)
		edge node {} (2v)
	(4u)
		edge node {} (2v)
		edge node {} (4v)
	(5u)
		edge node {} (1v)
		edge node {} (5v)
	;

\end{tikzpicture}
\end{center}
\caption{Incidence graph of decomposition $\cS_{Lader-\BZ_3^{std}}$}
\label{fig:lad_inc}
\end{figure}

\begin{figure}
\begin{center}
\begin{tikzpicture}[node distance=3cm and 1cm, auto]

	\node[main node, label={\small 4}] (1u) {$\begin{matrix} 1 & 0 & 0 \end{matrix}$};
	\node[main node, label={\small 4}] (2u) [right=of 1u] {$\begin{matrix} 0 & 1 & 0 \end{matrix}$};
	\node[main node, label={\small 4}] (3u) [right=of 2u] {$\begin{matrix} 0 & 0 & 1 \end{matrix}$};
	\node[main node, label={\small 3}] (4u) [right=of 3u] {$\begin{matrix} 1 & 1 & 1 \end{matrix}$};
	\node[main node, label={\small 2}] (5u) [right=of 4u] {$\begin{matrix} 0 & 1 & 1 \end{matrix}$};
	\node[main node, label={\small 1}] (6u) [right=of 5u] {$\begin{matrix} 1 & 1 & 0 \end{matrix}$};

	\node[main node, label=below:{\small 4}] (1v) [below=of 1u] {$\begin{matrix} 1 & 0 & 0 \end{matrix}$};
	\node[main node, label=below:{\small 4}] (2v) [below=of 2u] {$\begin{matrix} 0 & 0 & 1 \end{matrix}$};
	\node[main node, label=below:{\small 4}] (3v) [below=of 3u] {$\begin{matrix} 0 & 1 & -1 \end{matrix}$};
	\node[main node, label=below:{\small 3}] (4v) [below=of 4u] {$\begin{matrix} 1 & -1 & 0 \end{matrix}$};
	\node[main node, label=below:{\small 2}] (5v) [below=of 5u] {$\begin{matrix} 0 & 1 & 0 \end{matrix}$};
	\node[main node, label=below:{\small 1}] (6v) [below=of 6u] {$\begin{matrix} 1 & 0 & -1 \end{matrix}$};

	\path[line]
	(1u)
		edge node {} (2v)
		edge node {} (3v)
		edge node {} (5v)
	(2u)
		edge node {} (1v)
		edge node {} (2v)
		edge node {} (6v)
	(3u)
		edge node {} (1v)
		edge node {} (4v)
		edge node {} (5v)
	(4u)
		edge node {} (3v)
		edge node {} (4v)
		edge node {} (6v)
	(5u)
		edge node {} (1v)
		edge node {} (3v)
	(6u)
		edge node {} (2v)
		edge node {} (4v)
	;

\end{tikzpicture}
\end{center}
\caption{Incidence graph of decomposition $\cS_{2fix-\BZ_3}$}
\label{fig:twofix_inc}
\end{figure}
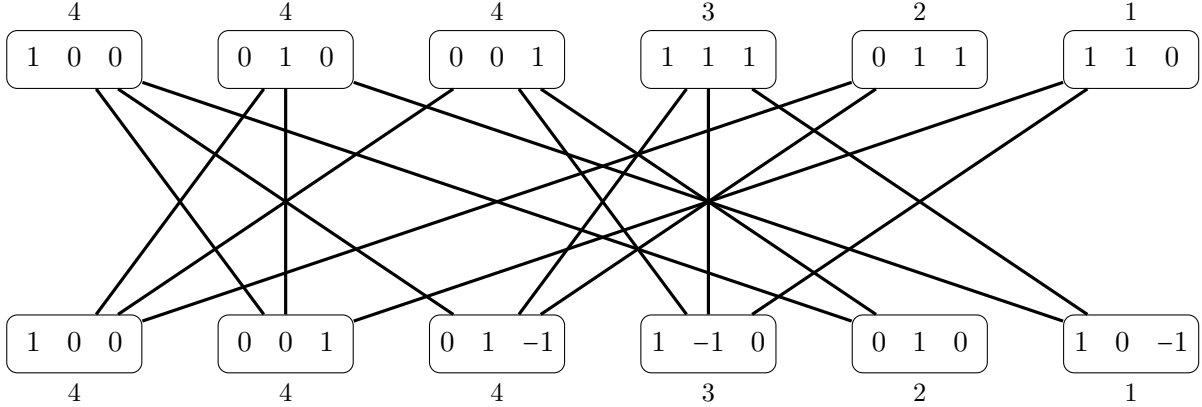

For decompositions where the three copies of $Mat_{n\times n}$ have been identified, such as our
$\BZ_3^{std}$-invariant decompositions, consider the {\it restricted family} $PGL_n\rtimes (\BZ_3\rtimes \BZ_2)\cdot \cS$, where
$PGL_n\subset PGL_n^{\times 3}$ is diagonally embedded and consider the corresponding restricted symmetry groups
$\G_{\cS}^{res} \subset PGL_n\rtimes (\BZ_3\rtimes \BZ_2)$. We define a second graph  $\cP\cG_{\cS}$ that
is an invariant of this restricted family, 
  the {\it pairing graph} which has an edge between $[\mu]$ and
$[v]$ if $\mu\ot v$ appears in the decomposition, and   triples of edges that appear in the same summands are grouped by color,
and one can weight the edge by the number of times it appears. 
Pairing graphs for the four decompositions are given in Figures \ref{fig:std_pair}, \ref{fig:zfzt_pair}, \ref{fig:lad_pair}, and \ref{fig:twofix_pair}.
The dashed black lines correspond to cubes, so should be interpreted as edges with multiplicity three.

As is clear from this discussion, one  can continue labeling and coloring to get
additional  information about the decomposition. 


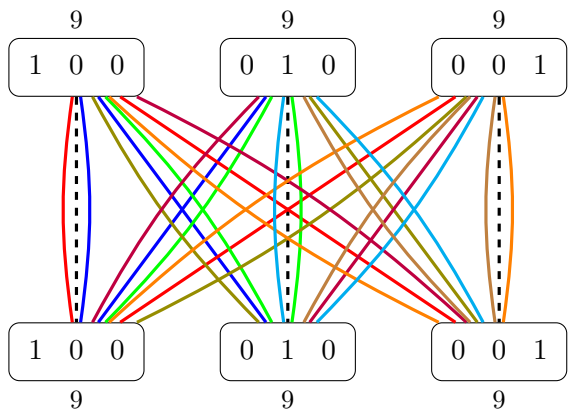
\begin{figure}
\begin{center}
\begin{tikzpicture}[node distance=3cm and 1cm, auto]

	\node[main node, label={\small 9}] (1u) {$\begin{matrix} 1 & 0 & 0 \end{matrix}$};
	\node[main node, label={\small 9}] (2u) [right=of 1u] {$\begin{matrix} 0 & 1 & 0 \end{matrix}$};
	\node[main node, label={\small 9}] (3u) [right=of 2u] {$\begin{matrix} 0 & 0 & 1 \end{matrix}$};

	\node[main node, label=below:{\small 9}] (1v) [below=of 1u] {$\begin{matrix} 1 & 0 & 0 \end{matrix}$};
	\node[main node, label=below:{\small 9}] (2v) [below=of 2u] {$\begin{matrix} 0 & 1 & 0 \end{matrix}$};
	\node[main node, label=below:{\small 9}] (3v) [below=of 3u] {$\begin{matrix} 0 & 0 & 1 \end{matrix}$};

	\path[line]
	(1u) edge [dashed,black,] node {} (1v)
	(2u) edge [dashed,black,] node {} (2v)
	(3u) edge [dashed,black,] node {} (3v)
	(2u) edge [blue,] node {} (1v)
	(1u) edge [blue,] node {} (2v)
	(1u) edge [blue,bend left=8] node {} (1v)
	(3u) edge [red,] node {} (1v)
	(1u) edge [red,] node {} (3v)
	(1u) edge [red,bend right=8] node {} (1v)
	(2u) edge [green,bend left=8] node {} (1v)
	(2u) edge [green,bend left=8] node {} (2v)
	(1u) edge [green,bend left=8] node {} (2v)
	(3u) edge [olive,bend left=8] node {} (1v)
	(2u) edge [olive,] node {} (3v)
	(1u) edge [olive,bend right=8] node {} (2v)
	(2u) edge [purple,bend right=8] node {} (1v)
	(3u) edge [purple,] node {} (2v)
	(1u) edge [purple,bend left=8] node {} (3v)
	(3u) edge [orange,bend right=8] node {} (1v)
	(3u) edge [orange,bend left=8] node {} (3v)
	(1u) edge [orange,bend right=8] node {} (3v)
	(3u) edge [cyan,bend left=8] node {} (2v)
	(2u) edge [cyan,bend left=8] node {} (3v)
	(2u) edge [cyan,bend right=8] node {} (2v)
	(3u) edge [brown,bend right=8] node {} (2v)
	(3u) edge [brown,bend right=8] node {} (3v)
	(2u) edge [brown,bend right=8] node {} (3v)
	;

\end{tikzpicture}
\end{center}
\caption{Pairing graph of standard decomposition}
\label{fig:std_pair}
\end{figure}

\begin{figure}
\begin{center}
\begin{tikzpicture}[node distance=3cm and 1cm, auto]

	\node[main node, label={\small 3}] (1u) {$\begin{matrix} 1 & 0 & 0 \end{matrix}$};
	\node[main node, label={\small 3}] (2u) [right=of 1u] {$\begin{matrix} 0 & 1 & 0 \end{matrix}$};
	\node[main node, label={\small 3}] (3u) [right=of 2u] {$\begin{matrix} 0 & 0 & 1 \end{matrix}$};
	\node[main node, label={\small 3}] (4u) [right=of 3u] {$\begin{matrix} 1 & 1 & 0 \end{matrix}$};
	\node[main node, label={\small 3}] (5u) [right=of 4u] {$\begin{matrix} 0 & 1 & 1 \end{matrix}$};
	\node[main node, label={\small 3}] (6u) [right=of 5u] {$\begin{matrix} 1 & 1 & 1 \end{matrix}$};

	\node[main node, label=below:{\small 4}] (1v) [below=of 1u] {$\begin{matrix} 1 & 0 & 0 \end{matrix}$};
	\node[main node, label=below:{\small 4}] (2v) [below=of 2u] {$\begin{matrix} 0 & 0 & 1 \end{matrix}$};
	\node[main node, label=below:{\small 4}] (3v) [below=of 3u] {$\begin{matrix} 1 & -1 & 0 \end{matrix}$};
	\node[main node, label=below:{\small 4}] (4v) [below=of 4u] {$\begin{matrix} 0 & 1 & -1 \end{matrix}$};
	\node[main node, label=below:{\small 1}] (5v) [below=of 5u] {$\begin{matrix} 0 & 1 & 0 \end{matrix}$};
	\node[main node, label=below:{\small 1}] (6v) [below=of 6u] {$\begin{matrix} 1 & 0 & -1 \end{matrix}$};

	\path[line]
	(4u) edge [dashed,black,] node {} (5v)
	(5u) edge [dashed,black,] node {} (6v)
	(1u) edge [dashed,black,] node {} (1v)
	(2u) edge [dashed,black,] node {} (3v)
	(3u) edge [dashed,black,] node {} (4v)
	(6u) edge [dashed,black,] node {} (2v)
	(2u) edge [blue,] node {} (2v)
	(5u) edge [blue,] node {} (4v)
	(3u) edge [blue,] node {} (3v)
	(3u) edge [red,] node {} (1v)
	(4u) edge [red,] node {} (2v)
	(6u) edge [red,] node {} (4v)
	(6u) edge [green,] node {} (3v)
	(5u) edge [green,] node {} (1v)
	(1u) edge [green,] node {} (2v)
	(1u) edge [olive,] node {} (4v)
	(4u) edge [olive,] node {} (3v)
	(2u) edge [olive,] node {} (1v)
	;

\end{tikzpicture}
\end{center}
\caption{Pairing graph of decomposition $\cS_{\BZ_4\times \BZ_3}$}
\label{fig:zfzt_pair}
\end{figure}
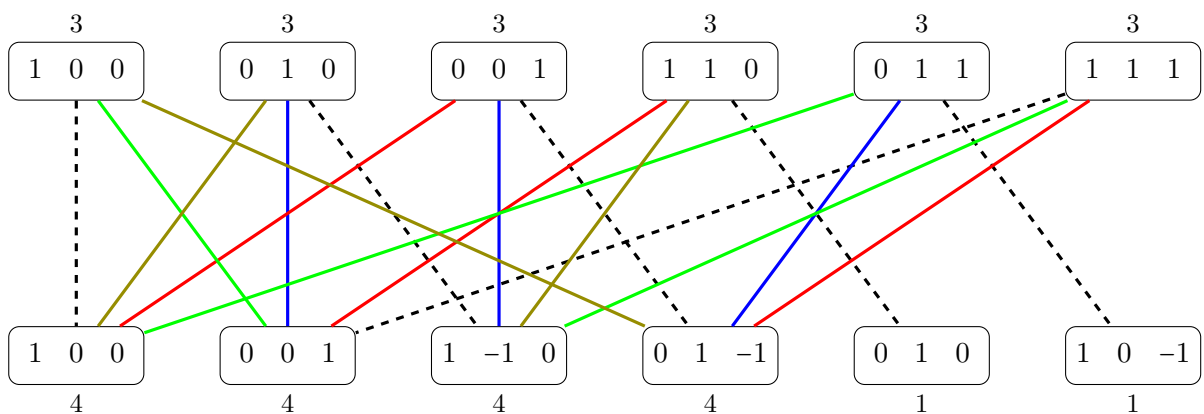

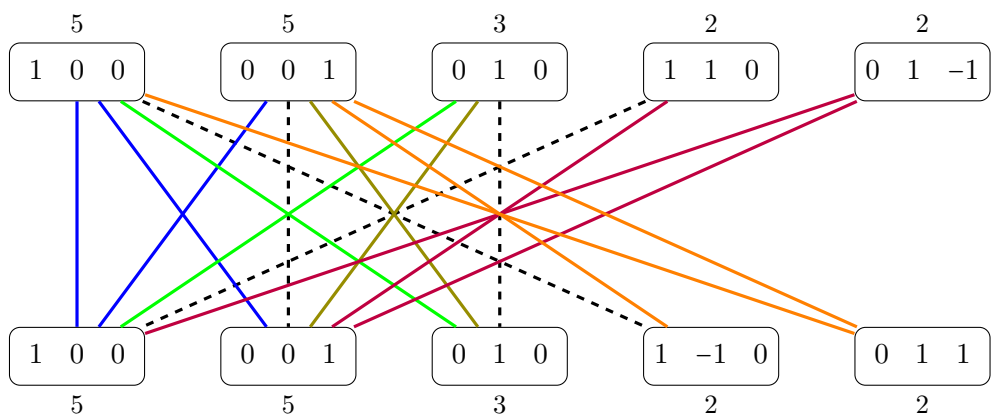
\begin{figure}
\begin{center}
\begin{tikzpicture}[node distance=3cm and 1cm, auto]

	\node[main node, label={\small 5}] (1u) {$\begin{matrix} 1 & 0 & 0 \end{matrix}$};
	\node[main node, label={\small 5}] (2u) [right=of 1u] {$\begin{matrix} 0 & 0 & 1 \end{matrix}$};
	\node[main node, label={\small 3}] (3u) [right=of 2u] {$\begin{matrix} 0 & 1 & 0 \end{matrix}$};
	\node[main node, label={\small 2}] (4u) [right=of 3u] {$\begin{matrix} 1 & 1 & 0 \end{matrix}$};
	\node[main node, label={\small 2}] (5u) [right=of 4u] {$\begin{matrix} 0 & 1 & -1 \end{matrix}$};

	\node[main node, label=below:{\small 5}] (1v) [below=of 1u] {$\begin{matrix} 1 & 0 & 0 \end{matrix}$};
	\node[main node, label=below:{\small 5}] (2v) [below=of 2u] {$\begin{matrix} 0 & 0 & 1 \end{matrix}$};
	\node[main node, label=below:{\small 3}] (3v) [below=of 3u] {$\begin{matrix} 0 & 1 & 0 \end{matrix}$};
	\node[main node, label=below:{\small 2}] (4v) [below=of 4u] {$\begin{matrix} 1 & -1 & 0 \end{matrix}$};
	\node[main node, label=below:{\small 2}] (5v) [below=of 5u] {$\begin{matrix} 0 & 1 & 1 \end{matrix}$};

	\path[line]
	(3u) edge [dashed,black,] node {} (3v)
	(2u) edge [dashed,black,] node {} (2v)
	(4u) edge [dashed,black,] node {} (1v)
	(1u) edge [dashed,black,] node {} (4v)
	(1u) edge [blue,] node {} (1v)
	(2u) edge [blue,] node {} (1v)
	(1u) edge [blue,] node {} (2v)
	(1u) edge [green,] node {} (3v)
	(3u) edge [green,] node {} (1v)
	(2u) edge [olive,] node {} (3v)
	(3u) edge [olive,] node {} (2v)
	(5u) edge [purple,] node {} (2v)
	(4u) edge [purple,] node {} (2v)
	(5u) edge [purple,] node {} (1v)
	(2u) edge [orange,] node {} (5v)
	(1u) edge [orange,] node {} (5v)
	(2u) edge [orange,] node {} (4v)
	;

\end{tikzpicture}
\end{center}
\caption{Pairing graph of decomposition $\cS_{Lader-\BZ_3^{std}}$}
\label{fig:lad_pair}
\end{figure}

\begin{figure}
\begin{center}
\begin{tikzpicture}[node distance=3cm and 1cm, auto]

	\node[main node, label={\small 4}] (1u) {$\begin{matrix} 1 & 0 & 0 \end{matrix}$};
	\node[main node, label={\small 4}] (2u) [right=of 1u] {$\begin{matrix} 0 & 1 & 0 \end{matrix}$};
	\node[main node, label={\small 4}] (3u) [right=of 2u] {$\begin{matrix} 0 & 0 & 1 \end{matrix}$};
	\node[main node, label={\small 3}] (4u) [right=of 3u] {$\begin{matrix} 1 & 1 & 1 \end{matrix}$};
	\node[main node, label={\small 2}] (5u) [right=of 4u] {$\begin{matrix} 0 & 1 & 1 \end{matrix}$};
	\node[main node, label={\small 1}] (6u) [right=of 5u] {$\begin{matrix} 1 & 1 & 0 \end{matrix}$};

	\node[main node, label=below:{\small 4}] (1v) [below=of 1u] {$\begin{matrix} 1 & 0 & 0 \end{matrix}$};
	\node[main node, label=below:{\small 4}] (2v) [below=of 2u] {$\begin{matrix} 0 & 0 & 1 \end{matrix}$};
	\node[main node, label=below:{\small 4}] (3v) [below=of 3u] {$\begin{matrix} 0 & 1 & -1 \end{matrix}$};
	\node[main node, label=below:{\small 3}] (4v) [below=of 4u] {$\begin{matrix} 1 & -1 & 0 \end{matrix}$};
	\node[main node, label=below:{\small 2}] (5v) [below=of 5u] {$\begin{matrix} 0 & 1 & 0 \end{matrix}$};
	\node[main node, label=below:{\small 1}] (6v) [below=of 6u] {$\begin{matrix} 1 & 0 & -1 \end{matrix}$};

	\path[line]
	(2u) edge [dashed,black,] node {} (3v)
	(2u) edge [blue,] node {} (6v)
	(5u) edge [blue,] node {} (3v)
	(1u) edge [green,] node {} (3v)
	(3u) edge [green,] node {} (1v)
	(3u) edge [green,] node {} (2v)
	(6u) edge [olive,] node {} (1v)
	(3u) edge [olive,] node {} (5v)
	(1u) edge [olive,bend left=8] node {} (3v)
	(1u) edge [purple,] node {} (1v)
	(3u) edge [purple,] node {} (4v)
	(4u) edge [purple,] node {} (2v)
	(2u) edge [orange,] node {} (2v)
	(2u) edge [orange,] node {} (4v)
	(4u) edge [orange,] node {} (5v)
	(5u) edge [cyan,] node {} (2v)
	(1u) edge [cyan,] node {} (4v)
	(4u) edge [cyan,] node {} (1v)
	;

\end{tikzpicture}
\end{center}
\caption{Pairing graph of decomposition $\cS_{2fix-\BZ_3}$}
\label{fig:twofix_pair}
\end{figure}
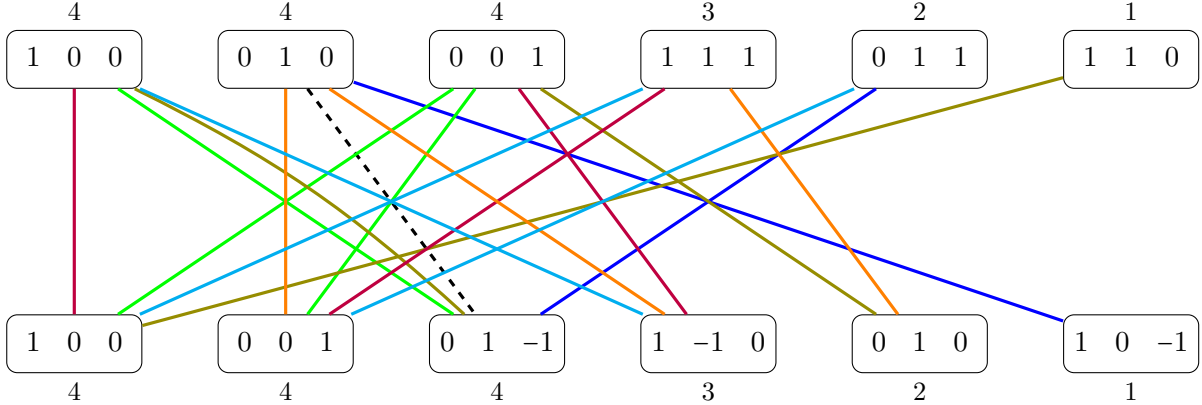

\section{Symmetry groups of our decompositions}\label{oursymgps}

With the graphs in hand, it is straightforward to determine the symmetry groups.

\subsection{Symmetries of $ {\cS_{\BZ_4\times \BZ_3}}$}

  \begin{proposition}\label{43symmetries} The symmetry group of $ {\cS_{\BZ_4\times \BZ_3}}$  is  $\G_{\cS_{\BZ_4\times \BZ_3}}=\BZ_4^{a_0}\times \BZ_3^{std}$, where  
$\BZ_4^{a_0}\subset PGL_3\subset PGL_3^{\times 3}$ is generated by $a_0$ of \eqref{aodef}.
  \end{proposition} 
\begin{proof}
The incidence graph shows no transpose-like symmetry is possible as points occur with different frequencies
in the two spaces.  Since we already know the $\BZ_3^{std}$ symmetry, we are reduced to determining
$\G_{\cS_{\BZ_4\times \BZ_3}}\cap PGL_3^{\times 3}$.

Say $\g=(g,h,k)\in PGL_3^{\times 3}$ preserves $\cS_{\BZ_4\times\BZ_3}$.
Since $a_0$ is is the unique term 
in the decomposition of full rank,  
$a_0^{\otimes 3}$ is fixed by $\g$. 
Since $a_0$ is fixed up to scale,  $[a_0] = [g a_0 h^{-1}] = [h a_0 k^{-1}] = [k
a_0 g^{-1}]$. This implies, up to a scale that we can ignore, that    $a_0^3 = g a_0 h^{-1} h a_0 k^{-1} k a_0
g^{-1} = g a_0^3 g^{-1}$. Hence $a_0^3$ and $g$ commute. 
(We will argue similarly several times in what follows.) 
Two matrices $A,B$,
with $B$ invertible,  commute if and only if $A$ commutes with $B^{-1}$. Since $a_0^3 = a_0^{-1}$,  we have $g, g^{-1}$ commute  with $a_0$. This reasoning holds for $h$ and $k$ as well.

Since $g, h, k$ all commute with $a_0$,   $\g$ commutes with the $\Z_4$ action.  
Since $\g$ preserves rank,   the   orbit \eqref{43d} must be fixed. 
Since $\g$ commutes with the $\Z_4$ action,  it suffices to determine it
 up to a power of the $\Z_4$ action. Namely we can expect $\g$ to fix one of the elements of the $4$ orbits of rank 2 matrices,
 e.g., the matrix in \eqref{43d}
\be\label{labelM}M= \begin{pmatrix} 0&- 1&0\\ 1&-1&0\\ 0&0&0\end{pmatrix}.
\ene

Case 1:   $M = g M h^{-1} = h M k^{-1} = k M g^{-1}$, which implies $g, h, k$ all commute with $M^3$ (using the same reasoning as above). The only matrices which commute  with both $M^3$ and $a_0$ are scalar multiples of the identity and we conclude.

Case 2: $a_0Ma_0\inv=g M h^{-1} = h M k^{-1} = k M g^{-1}$ which implies $a_0M^3a_0\inv=gM^3g$ so $a_0\inv gM^3=(a_0\inv g)\inv$.
Hence $a_0\inv g$ commutes with $M^3$ and $a_0$, thus is a scalar multiple of the identity and
$g=h=k=a_0$, which has already been accounted for.  

The other two cases, like case 2, only with $a_0^2$ and $a_0^3$ playing the role of $a_0$,  are similar.  
\end{proof}

\begin{proposition}\label{32symprop} In  the family of decompositions
$PGL_3^{\times 3}\rtimes (\BZ_3\rtimes \BZ_2) \cdot  {\cS_{\BZ_4\times \BZ_3}}$,
the set of  $ \BZ_3^{std}$-invariant decompositions
is  the  image of the diagonal $PGL_3$-action on $\cS_{\BZ_4\times \BZ_3}$ times the standard transpose $\BZ_2$.
\end{proposition}
\begin{proof}
We need to show any  $\g=(g,h,k)\in PGL_3^{\times 3} $ that takes $\cS_{\BZ_4\times \BZ_3}$
 to another decomposition invariant under the same standard $ \BZ_3$ satisfies $g=h=k$.
 
Since $\g$ fixes rank and there is only one tensor consisting of
 three matrices of rank 3,   $\g\cdot a_0^{\otimes 3}$ must be a $\Z_3$-fixed point.
 Thus  $ g a_0 h^{-1} = h a_0 k^{-1} = k a_0 g^{-1}$. Write $a= g a_0 h^{-1}$.    This implies $a^3 = g a_0^3 g^{-1} = h a_0^3 h^{-1} = k a_0^3 k^{-1}$. Equivalently $a^{-3} = g a_0 g^{-1}$, and the same for $h, k$. Thus
 $$a = a^{-3} g h^{-1} = g h^{-1} a^{-3}.
 $$
Equivalently
$$a^4 = g h^{-1}.
$$

We have $g^{-1} a g$ is a cube root of $a_0^{-1}$.
Recall that  $a_0^{-1}$ has distinct eigenvalues $-1, i, -i$ so all cube roots must likewise have distinct eigenvalues. This uniquely determines the finite choices of eigensystems for our cube roots,  yielding 27 possibilities. 
The $\Z_3^{std}$ action commutes with the diagonal $PGL_3$, 
so  we can assume $g = \Id$.
Then $a$ is a cube root of $a_0^{-1}$, so $a$ has order $12$. Combining this we the observations $a^4 = h^{-1} = h k^{-1}$ shows
$h = a^{-4}$ and
$k = a^{-8}$. 
This leaves us with 27 total candidate restricted families  that have a rank three $\BZ_4\times \BZ_3$-fixed summand.

Since $g, h, k$ all commute with $a_0 = a^{-3}$,    $\g$  commutes with the $\Z_4$ action. 
Consider the action of $\g$
on the  two terms in \eqref{43c}.  
Since 
the $\Z_4$ action sends $\Z_3$ fixed points to $\Z_3$ fixed points,  these summands
each  map to $\Z_3$ fixed points.

For $\g$ to send $M$ of \eqref{labelM} to another $\Z_3$ fixed point we must have
$$M a^4 = a^{-4} M a^{8} = a^{-8} M = a^4 M.$$

Testing all $27$ cube roots of $a_0^{-1}$ we see that $M$ commuting with $a^{-4}$
implies $a$ is   a scalar times  $a_0$, and we may assume $a = a_0$, so $a^4 = \Id$. 
This implies that there are no other   $  \BZ_3^{std}$-invariant  subfamilies in the family.
\end{proof}

 \subsection{Symmetries of $\cS_{Lader-\BZ_3}$}

 The symmetry group of $G_{\cS_{Lader-\BZ_3}}$ has already been discussed. 
 
 \begin{proposition}  In Laderman the family of decompositions,
the set of  $ \BZ_3^{std}$-invariant decompositions
is  the  image of the diagonal $PGL_3$-action on $\cS_{Lader-\BZ_3}$ times the standard transpose $\BZ_2$.
\end{proposition}

The proof is very similar to that of Proposition \ref{32symprop}, so is omitted.

\subsection{Symmetries of  $\cS_{2fix-\BZ_3}$ }


\begin{proposition} The symmetry group of $\cS_{2fix-\BZ_3}$ is $\G_{ \cS_{2fix-\BZ_3}}=\BZ_3$.
\end{proposition}
\begin{proof}
The incidence and pairing graphs have no joint automorphisms, which  shows there are no additional diagonal $GL_3\subset GL_3^{\times 3}$ symmetries.
 
To show, despite the symmetry of the graphs, that  there is no transpose-like symmetry, i.e., a symmetry of the form
$x\ot y\ot z\mapsto gx^Tk\inv \ot kz^Th\inv \ot hy^Tg\inv$, first note that
  the first matrix in \eqref{twofix_c3}, call this $M_{14a}$, must satisfy
$g M_{14a}^T    = M_{14a}h$ as \eqref{twofix_c3} is the only triple with ranks $(2,1,1)$, and similarly
(using the $\BZ_3$-action), $h M_{14a}^T   = M_{14a}k$ and $k M_{14a}^T   = M_{14a}g$.
Moreover the second and third matrices in this triple, call them $M_{14b},M_{14c}$ must satisfy
$hM_{14c}^Tk\inv = M_{14b}$ and $kM_{14b}^Tg\inv=M_{14c}$, which forces
\be\label{ghk}
g=h=k=\begin{pmatrix} s&s&s\\ s&s&0\\ s&0&0\end{pmatrix}
\ene
and we may normalize $s=1$.
Now apply this to the triple in \eqref{twofix_c4}, we get the triple
$$
\begin{pmatrix}0&1&0\\0&0&1\\ 0&0&1\end{pmatrix}\ot \begin{pmatrix}1&0&0\\0&0&1\\ 0&0&0\end{pmatrix}
\ot \begin{pmatrix}0&0&0\\1&0&-1\\ 0&1&-1\end{pmatrix}
$$
which is not a triple appearing in the decomposition.

It remains to show there are no additional symmetries coming from $GL_3^{\times 3}$.
There are only two $\Z_3$ fixed points in $\cS_{2fix-\BZ_3}$, $M_{12}$ appearing in \eqref{twofix_c1} of rank one and $M_{13}$ appearing
in \eqref{twofix_c2}, so these must be fixed by 
any $(g,h,k)\in PGL(U)\times PGL(V)\times PGL(W)$-symmetry.  
Since both of these matrices are idempotent, we get, by an argument as in the proof of Proposition \ref{32symprop}, 
\begin{align}\label{mtwid} M_{12}& = M_{12}^3 = g M_{12}^3 g^{-1} = gM_{12} g^{-1},\\
\nonumber
 M_{13}& = M_{13}^3 = g M_{13}^3 g^{-1} = g M_{13} g^{-1}.
 \end{align}

There is only one tensor with ranks $(2,1,1)$ in the decomposition, so this also must be fixed. Since it is $\Z_3$ invariant, we get the $(1, 2, 1)$ tensor and the $(1,1,2)$ tensor also must be fixed. These matrices are $M_{14a},M_{14b},M_{14c}$.
Then $M_{14a} = g M_{14a} h^{-1}$, $M_{14b} = h M_{14b} k^{-1}$, $M_{14c} = k M_{14c} g^{-1}$. Combining this we get $M_{14a},M_{14b},M_{14c}$ commutes with $g$.  

Finally we  check that the only matrix which commutes with $ M_{12},M_{13}, M_{14a},M_{14b},M_{14c}$  is the identity matrix.
\end{proof}

\begin{remark}  Remarkably, the group element \eqref{ghk}   splits $\Mthree$ into the sum
of two tensors: $T_1$, the sum of \eqref{twofix_c1},\eqref{twofix_c2},\eqref{twofix_c3},\eqref{twofix_c5},\eqref{twofix_c6} which {\it is} invariant under
the transpose like action, and  the sum of  the others, call it $T_2$, which is sent to a different decomposition of $T_2$ under the action.
\end{remark}
 
\begin{proposition} In  the family of decompositions
$PGL_3^{\times 3}\rtimes (\BZ_3\rtimes \BZ_2) \cdot  \cS_{2fix-\BZ_3}$,
the set of  $ \BZ_3^{std}$-invariant decompositions
is  the  image of the diagonal $PGL_3$-action on $\cS_{2fix-\BZ_3}$ times the standard transpose $\BZ_2$.
\end{proposition}
\begin{proof}
The  two $\Z_3$ fixed points $M_{12},M_{13}$ must be fixed by any symmetry, and as above, for any triple $(g,h,k)$,  equation  \eqref{mtwid} still holds
(in fact for $g$ replaced by $h$ or $k$ as well).

Similarly the triple \eqref{twofix_c3} and its $\BZ_3$-translates also must be fixed.  
Then $M_{14a} = g M_{14a} h^{-1}$, $M_{14b} = h M_{14b} k^{-1}$, $M_{14c} = k M_{14c} g^{-1}$. Combining this we get $M_{14a}M_{14b}M_{14c}$ commutes with $g$. We conclude
as above. 
 \end{proof}

\section{Configurations of points in projective space}\label{configsect}
In the decompositions,   
vectors   appear 
  tensored with   other vectors, so they are really only
  defined up to scale  (there is only a \lq\lq global scale\rq\rq\ for each term).
  This suggests using points in projective space, and only later taking scales into account. 

Towards our goal of building new decompositions, we would like to 
describe existing decompositions in terms of   simple building blocks.
The standard cyclic  $\BZ_3$ invariant decompositions of $\Mn$  naturally come in
the restricted families parameterized
by the diagonal $PGL_n\subset PGL_n^{\times 3}$.
Thus, when we examine, e.g., the points in $\BP U=\pp{n-1}$ appearing in the rank
one terms in  a decomposition, we should
really study the set of points up to $PGL_n$-equivalence, call such a {\it configuration}.
Identifying configurations will also facilitate comparisons between known decompositions. 

The simplest configuration is $n$ points in $\BC^n$ that form a basis, as occurs with
the standard decomposition.   All known   decompositions
of size less than $n^3$  use more than
$n$ points. The next simplest is a collection of $n+1$ points in general linear position, i.e., 
a collection of points such that any subset of $n$ of them  forms a basis.
Call such a {\it framing} of $\pp{n-1}$.  Just as all bases
are $PGL_n$-equivalent to the standard basis, all  framings,
as points in projective space,  are equivalent to:
$$
\begin{pmatrix} 1\\ 0 \\ 0\\ \vdots \\0\end{pmatrix},\
\begin{pmatrix} 0\\ 1 \\ 0 \\  \vdots \\0\end{pmatrix},\cdots,
\begin{pmatrix} 0\\ 0 \\    \vdots \\0 \\ 1\end{pmatrix}, \
\begin{pmatrix}  1\\ 1 \\    \vdots \\1 \\ 1\end{pmatrix}.
$$

We focus on the case of  $\pp 2$.

\subsection{The case of   $\pp 2$}\label{defaultpinning} The group   $PGL_3$  acts simply transitively  on the set of   $4$-ples of points  in $\pp 2$  in general linear position
(i.e.\ such that  vectors associated to any three of them form a basis).

Start with any $4$-ple of points in general linear position. 
We will call the following choice, the {\it default framing}:
$$
u_{1} =\begin{pmatrix} 1\\ 0\\ 0\end{pmatrix}, \ 
u_{2}=\begin{pmatrix} 0\\ 1\\ 0\end{pmatrix}, \ 
u_{3}=\begin{pmatrix} 0\\ 0\\ 1\end{pmatrix}, \ 
u_{4}=\begin{pmatrix} -1\\  -1\\ -1\end{pmatrix}. 
$$
Note that $u_1,u_2,u_3$ is the standard basis and $u_4$ is chosen such that $u_1+u_2+u_3+u_4=0$.  

The $\{ [u_j]\}$  determine $6$ lines in $\BP U$, those going through
pairs of points,  that we consider as points in $\BP U^*$.

For the default framing, representatives of these are:
\begin{align*}
&v_{12} =(0,0,1), \ \ v_{13} =(0,1,0), \ \ v_{14} =(0, 1,-1),\\
& v_{23} =(-1,0,0), \ \ v_{24} =(-1,0,1), \ \ v_{34} =(1,-1,0).
\end{align*}
Here $[v_{i j}]$ is the line in $\pp 2$, considered as a point in the dual space $\pp{2*}$,  through the points $[u_i]$ and $[u_j]$ in $\pp 2$  (or dually, the point of intersection
of the two lines $[u_i]$, $[u_j]$ in $\pp{2*}$).  Algebraically this means $v_{ij}(u_i)=0=v_{ij}(u_j)$.  

 The choices of scale made  here are useful for the decomposition $\cS_{\BZ_4\times \BZ_3}$  
because they make the $\BZ_4$ action easier to write down. 
They are such    that $v_{i,i+1}(u_{i+2})=1$, $v_{i,i+1}(u_{i+3})=-1$ (indices considered mod four).
This has the advantage of $v_{i+1,i+2}=a_0^{-i} v_{12}$ where $a_0$ is as in \eqref{aodef}. 
For $v_{13}$  there was no obvious choice of sign, and  we chose $v_{24}={a_0}\inv  (v_{13})$.

The $v_{ij}$'s constitute   two $\BZ_4$-orbits: the $v_{i,i+1}$'s
which consist of four vectors, and the $v_{i,i+2}$'s of which there are two.

The $v_{i,j}$ in turn determine their new points of intersection:
$$
u_{12,34} =\begin{pmatrix} 1\\ 1\\ 0\end{pmatrix}, \ 
u_{13,24}=\begin{pmatrix} 1\\ 0 \\ 1\end{pmatrix}, \ 
u_{14,23}=\begin{pmatrix} 0\\ 1\\ 1\end{pmatrix}.
$$

\begin{figure}[!htb]\begin{center}\label{basisconf}
\includegraphics[scale=.4]{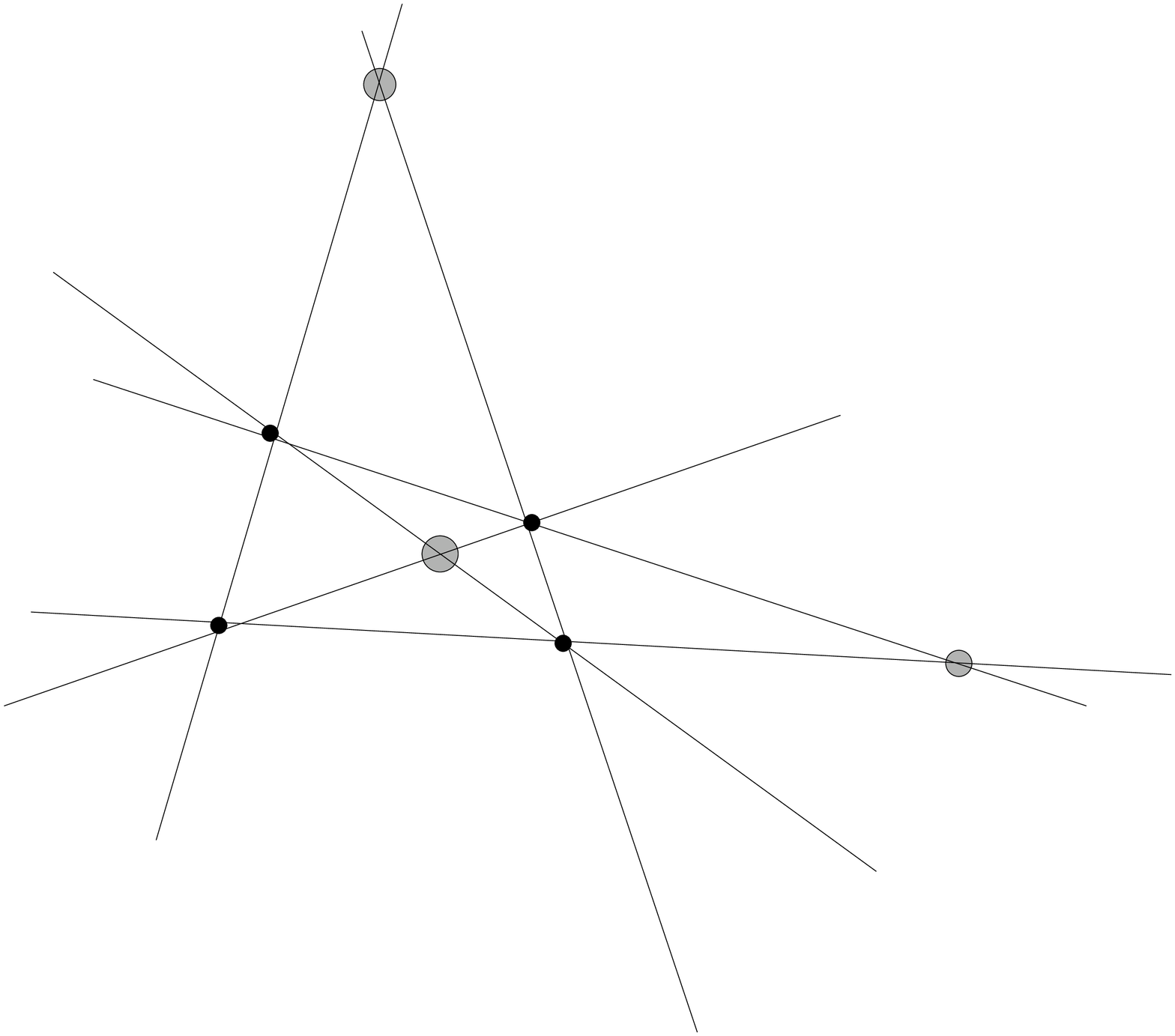}
\caption{\small{Points and lines from the default framing of $\pp 2$. Original four points are black dots, three new are shaded dots}}  
\end{center}
\end{figure}

These  determine further lines 
$$
v_{(12,34),(13,24)} =(-1, 1,1), \ \    v_{(12,34),(14,23)} =(1,-1,1), \ \ v_{(13,24),(14,23)} =(1,1,-1),
$$
which determine
$$
u_{34,(13,24|14,23)} =\begin{pmatrix} 1\\ 1\\ 2\end{pmatrix}, \ 
u_{24,(12,34|14,23)}=\begin{pmatrix} 1\\ 2 \\ 1\end{pmatrix}, \ 
u_{23,(12,34|13,24)}=\begin{pmatrix} 2\\ 1\\ 1\end{pmatrix},
$$
$$
u_{12,(12,34|13,24) } =\begin{pmatrix} 1\\  1\\ 0\end{pmatrix}, \ 
u_{13,(12,34|13,24) }=\begin{pmatrix} 1\\ 0 \\  1\end{pmatrix}, \ 
u_{14,(12,34|13,24) }=\begin{pmatrix} 0\\ 1\\  1\end{pmatrix}.
$$
This process continues, giving rise to an infinite collection of points  but in practice only vectors from the first $3$
rounds appeared in decompositions.

\subsection{Point-line configuration for $\cS_{\BZ_4\times \BZ_3}$}

The rank one elements appearing in $\cS_{\BZ_4\times \BZ_3}$ consist of points from
three rounds of points obtained from the default configuration. All points appear except 
that $u_{13|24}$ is missing (the orbit under $\BZ_4$ of $u_{12|34}$ is
$\{ u_{12|34}, u_{14|23}\}$).

Here is the decomposition $\cS_{\BZ_4\times \BZ_3}$ in terms of the  points  from \S\ref{defaultpinning}:

\begin{align}
\label{aa1x} \Mthree =& -a_0^{\ot 3} \\
\label{aa2x} &+\BZ_4\cdot (u_{1}  v_{23})^{\ot 3}\\
\label{aa3x}&+ \BZ_2 \cdot (u_2v_{23}-u_{12|34}v_{23})^{\ot 3}  \\
\label{aa6x} &+\BZ_4\cdot  ( u_{2}  v_{24})^{\ot 3} \\
\label{aa15x} &  +
\BZ_3\times \BZ_4\cdot (  u_{2}  v_{23}\ot  u_{4}  v_{14}\ot u_{12|34}  v_{13}).
\end{align}

\subsection{Point-line configuration for $\cS_{Lad-\BZ_3}$}
Thanks to the transpose-like  symmetry, it is better to label points in the dual  space
by their image under transpose rather than annihilators, to make the transpose-like symmetry
more transparent. Points:
$$
u_1=\begin{pmatrix} 1\\ 0\\ 0\end{pmatrix}, \ u_2=\begin{pmatrix} 0\\ 1\\ 0\end{pmatrix}, \ 
u_3=\begin{pmatrix} 0\\ 0\\ 1\end{pmatrix}, \ 
u_{12}=\begin{pmatrix} 1\\ -1\\ 0\end{pmatrix}, \ u_{23}=\begin{pmatrix} 0\\ 1\\ -1\end{pmatrix}.
$$
\begin{align*}
&v_{1}=(1,0,0), \ v_{2}=(0,1,0), \ v_{3}=(0,0,1), \\ 
&v_{12}=(1,1,0), \ v_{23}=(0,1,1).
\end{align*}

This collection of points  has a $\BZ_2$-symmetry generated by $\t_{13}$ which swaps the two lines.

\begin{figure}[!htb]\begin{center}\label{laderpic}
\includegraphics[scale=.4]{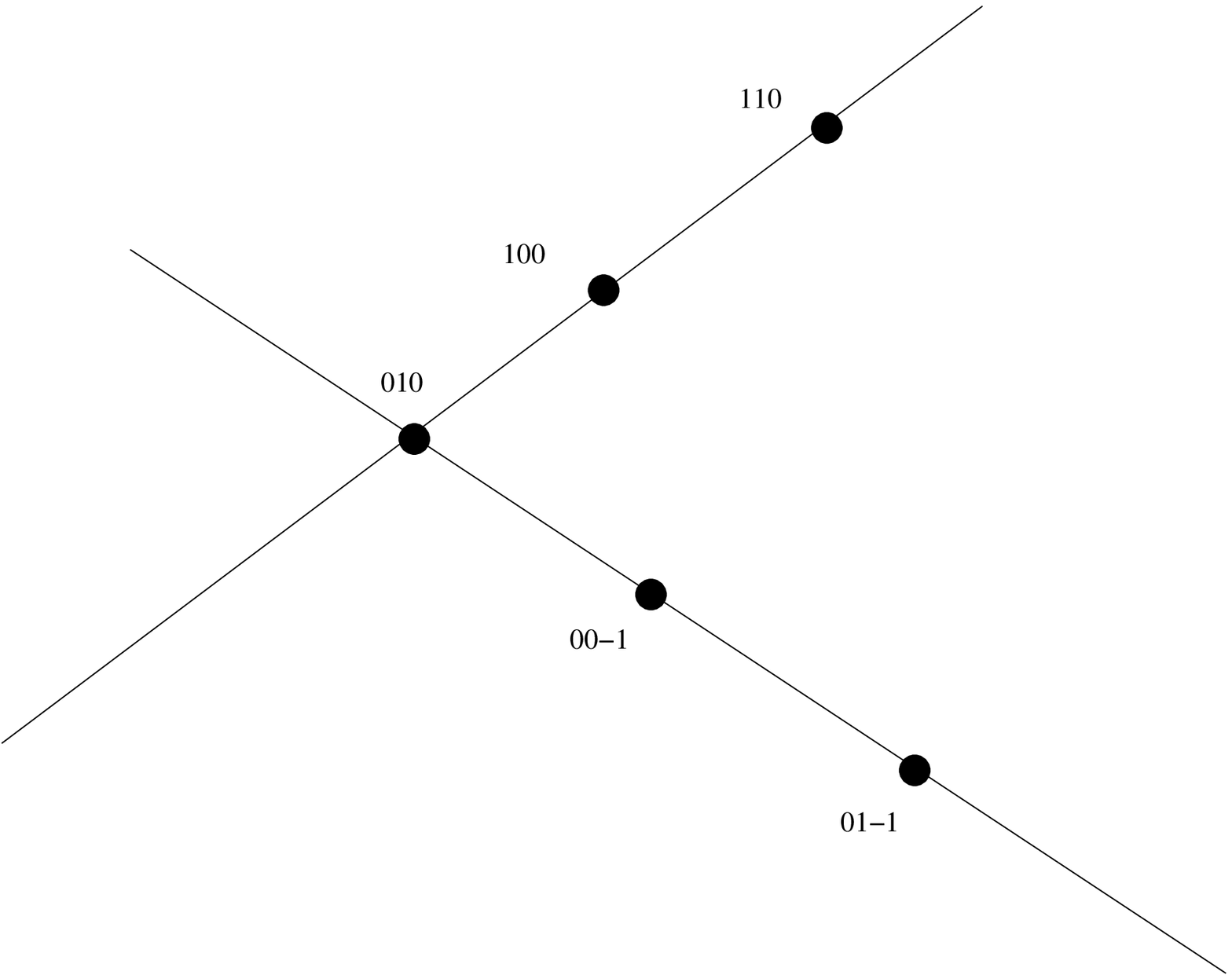}
\caption{\small{Rank one $\BP U$-points appearing in $\cS_{Lad-\BZ_3}$}}  
\end{center}
\end{figure}

In order to express $\cS_{Lad-\BZ_3}$ in terms of just these points, we   write the decomposition in
terms of the $\BZ_3\rtimes \BZ_2^{\zeta}$-orbits.

\begin{align}\Mthree =
\label{l1}  &   ( u_2  v_{2})^{\ot 3}\\
\label{l2} & + ( u_3  v_{3})^{\ot 3} \\
\label{l3} &+(u_{12}   v_{1})^{\ot 3} \\
\label{l4} &+( u_1 v_{12})^{\ot 3}\\
\label{l5} &+ ( u_{2}  v_{1}-u_1  v_{12} )^{\ot 3}
\\
\label{l6} & +
\BZ_3\cdot  (u_{1 }  v_{3})
\ot (u_{3}  v_{1})
\ot (u_{1}  v_{1}) 
\\
\label{l7} &+
\BZ_3\cdot  (u_{23 }  v_{1})
\ot (u_{12}  v_{3})
\ot (u_{23}  v_{3}) 
\\
\label{l8} &+\BZ_3\cdot  
(u_{3}  v_{12})
\ot (u_{1}  v_{23})
\ot ( u_{3}  v_{23}) 
\\
\label{l9} &+\BZ_3\cdot  ( u_{2}  v_{3}-u_{23}  v_{1} )\ot 
( u_{1} v_{2}-u_{12}  v_{3})\ot 
( u_{3 }  v_{2}-u_{23}  v_{3}) \\
\label{l10} & +\BZ_3\rtimes \BZ_2^{\zeta}\cdot  (u_{23}   v_{12} +u_2  v_3- u_1 v_{23} )
\ot (u_{2}  v_{3})
\ot (u_{3}  v_{2})   .
\end{align}

\subsection{Point-line arrangement for $\cS_{2fix-\BZ_3}$}

 Despite  the lack of a transpose-like symmetry, both sets of points
are the $6$ points corresponding to lines dual to the standard configuration of four points.
Again, this illustrates how the decomposition nearly has such a symmetry, and could likely be modified to have such.

\begin{figure}[!htb]\begin{center}
\includegraphics[scale=.4]{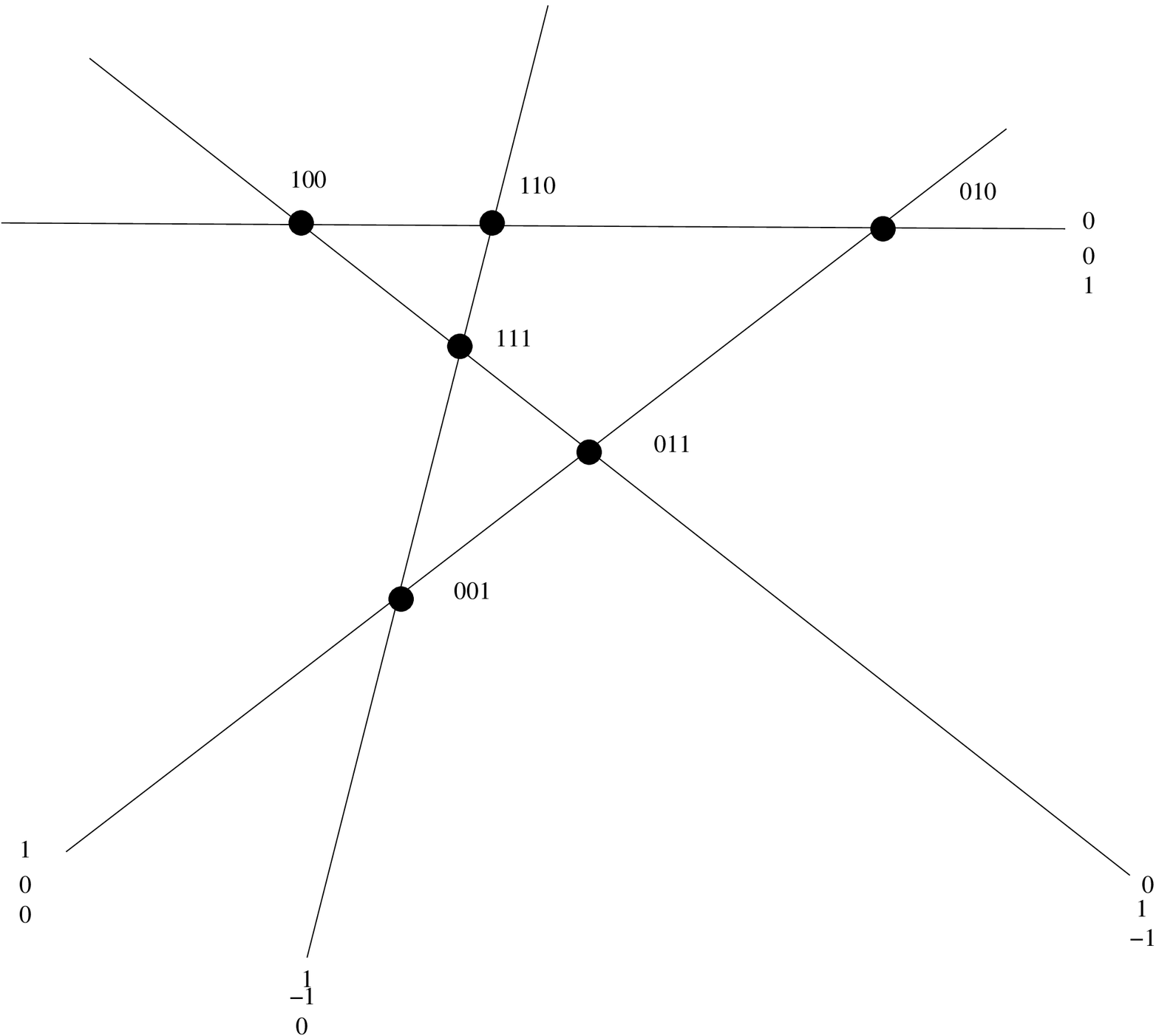}
\caption{\small{Points from $\cS_{2fix-\BZ_3}$ rank one matrices appearing in $\BP U$ and $\BP U^*$, with the latter expressed as lines in $\BP U$}}  
\end{center}
\end{figure}

\section{Spaces of $\G$-invariants for various $\G$}\label{repsect}
The space of $\BZ_3$-invariants in $A^{\ot 3}$ is $S^3A\op \La 3 A$.  
To see this, as a  $GL(A)\times \FS_3$-module we have the decomposition
$A^{\ot 3}=S^3A\ot [3]\op S_{21}A\ot [21]\op \La 3 A\ot [111]$.
Then $\BZ_3$ acts trivially on the trivial representation $[3]$ and trivially
on the sign representation $[111]$ (as the three cycle is even), and
$[21]$ decomposes into its $e^{\frac{2\pi i}3},e^{\frac{4\pi i}3}$ eigenspaces under $\BZ_3$. 
Thus, if $\tdim A=m$,   the space of invariants $(A^{\ot 3})^{\BZ_3}$ has dimension $\frac 13(m^3+2m)$, so restricting the search to $\BZ_3$-invariant
decompositions reduces the search size by about a factor of $3$. (In our case, $m=n^2$.)

The $\BZ_3$-fixed triples all lie in $S^3A$, so one cannot have all terms of a decomposition
individually $\BZ_3$-fixed.
{\it A-priori} there could be $2+3k$ $\BZ_3$-fixed points for $0\leq k<7$. We found decompositions with $2,5$, and $11$ $\BZ_3$-fixed points,
i.e., $k=0,1,3$.

\subsection{$\BZ_{n+1}$-invariants}
Consider the diagonal $\BZ_{n+1}\subset PGL_n\subset PGL_n^{\times 3}$ invariants: Each of $U^*,V$ decomposes into  $n$ one-dimensional representations for
$\BZ_{n+1}$,  corresponding to the  eigenvalues $\o,\o^2\hd \o^{n}$ where $\o =e^{\frac{2\pi i}{n+1}}$.
Write $(U^*)_j,V_j$ for the eigenspace corresponding to $\o^j$.
Then, adding indices modulo $n+1$, and letting $u\in [n]$, 
\begin{align*}
A_0=\bigoplus_{t=1}^n U^*_t\ot V_{n+1-t},\\
A_u=\bigoplus_{i=1, i\neq u}^n U^*_i\ot V_{u-i},
\end{align*}
so $\tdim A_0=n$, and   $\tdim A_u=n-1$,
where each $A_i$ is a $\BZ_{n+1}$-isotypic component of~$A$.
Note that $A_0$ is spanned by the powers of $a_0$ of \eqref{aodef}.
 
The space of $\BZ_{n+1}$ invariants in   $A\ot B\ot C$  is spanned by the spaces  $A_{\a}\ot B_{\b}\ot C_{\g}$ with $\a+\b+\g\equiv 0\tmod n+1$.
Let $1\leq i,j\leq n$,  we have the following dimensions:
\begin{center}\begin{tabular}{c|c|c|c}  
space & dimension & number  of   such& total   contribution\\
\hline
$A_{0}\ot B_{0}\ot C_0$ &\ \  $n^3$ & $1$ &\ \  $n^3$ \\
$A_{0}\ot B_{j}\ot C_{n+1-j}$  plus cyclic perms  &$n(n-1)^2$ & $3n$ &\ \   $3n^2(n-1)^2$\\
$A_{i}\ot B_{j}\ot C_{n+1-j-i}$,\    $i+j\neq n+1$ &\ \   $(n-1)^3$ & $n(n-1)$ &\ \   $n(n-1)^4$
\end{tabular}\end{center}
 
Thus 
\begin{proposition} The space of   $\BZ_{n+1}\subset PGL_n\subset PGL_n^{\times 3}$ invariants for the diagonal $\BZ_{n+1}$  in $\BC^{n^2}
\ot \BC^{n^2}\ot \BC^{n^2}$
is $n^5- n^4+ n^3- n^2+ n$. In particular, when $n=2,3,4$ these dimensions are respectively $22,183,820$.
\end{proposition}

\subsection{ $\BZ_{n+1}\times \BZ_3$-invariants} We look inside $S^3A\op \La 3A$ for $\BZ_{n+1}$-invariants.
Write
%
\begin{align*}
S^3A&=S^3(A_0+\cdots +A_n) = S^3A_0 \quad\op\quad \bigoplus_{j=1}^n S^3A_j \quad\op\quad \bigoplus_{j=1}^n A_0\ot S^2A_j \quad\op\quad \bigoplus_{j=1}^n A_j\ot S^2A_0\\
&\op\quad \bigoplus_{j,k=1, j \neq k}^n A_j\ot S^2A_k \quad\op\quad \bigoplus_{1 \leq j<k\leq n} A_0\ot A_j \otimes A_k \quad\op\quad \bigoplus_{1\leq i<j<k\leq n}A_i\ot A_j\ot A_k.
\end{align*}
The subspace of $\BZ_{n+1}$-invariants is (here all indices run from $1$ to $n$)
\begin{align*}
& (S^3A)^{\BZ_{n+1}} = S^3A_0 \quad\op\quad \bigoplus_{j\mid\,   3j\equiv 0 \tmod n+1 }  S^3A_j \quad\op\quad
\bigoplus_{j\mid\,   2j\equiv 0 \tmod n+1 } A_0\ot S^2A_j \quad\op\quad 0 \\
&\op
\bigoplus_{(j,k)\mid\,   j+2k\equiv 0 \tmod n+1; j\neq k} A_j\ot S^2A_k
 \quad\op\quad
\bigoplus_{i\mid\,  i<n+1-i} A_0\ot A_i\ot A_{n+1-i}
\quad\op\quad \bigoplus_{(i,j)\mid\,  i<j<n+1-i-j \atop \text{ or } i<j<2n+2-i-j<n+1}A_i\ot A_j\ot A_{n+1-i-j}.
\end{align*}

The dimensions of the summands of the various types are respectively
$$
\binom{n+2}3,\ \binom{n+1}3,\ n\binom{n}2,\ 0 ,\ (n-1)\binom{n}2,\ n(n-1)^2,\ (n-1)^3 .
$$

 Similarly, the space of $\BZ_{n+1}$-invariants in  $\La 3 A$ is
\begin{align*}
(\La 3A)^{\BZ_{n+1}} =&\La 3A_0\op \bigoplus_{j\mid\,   3j\equiv 0 \tmod n+1 }  \La 3A_j \op 
\bigoplus_{j\mid\,   2j\equiv 0 \tmod n+1 } A_0\ot \La 2A_j 
 \op 0 \\
 &\op
\bigoplus_{(i,j)\mid\,   i+2j\equiv 0 \tmod n+1; i \neq j} A_i\ot \La 2A_j
\op
\bigoplus_{i\mid\,  i<n+1-i} A_0\ot A_i\ot A_{n+1-i}
\op \bigoplus_{(i,j)\mid\,  i<j<n+1-i-j \atop \text{ or } i<j<2n+2-i-j<n+1}A_i\ot A_j\ot A_{n+1-i-j}.
\end{align*}
The dimensions of the summands of the various types are respectively
$$
\binom{n }3,\ \binom{n-1}3,\ n\binom{n-1}2,\ 0,\ (n-1)\binom{n-1}2,\ n(n-1)^2,\ (n-1)^3 .
$$

In both cases, 
the number of terms of each type depends on divisibility properties of $n$.
When $n=2$, the summands are 
$$
S^3A_0\op S^3A_1\op S^3A_2\op (A_0\ot A_1\ot A_2)^S \op (A_0\ot A_1\ot A_2)^{\Lambda},
$$
for a total dimension of $8+2=10$. (Here the $S,\Lambda$ superscripts are whether the factor is in $S^3$ or $\La 3$.)

When $n=3$, the summands are
\begin{align*}
&S^3A_0\op A_0\ot S^2A_2\op  A_2\ot S^2A_1\op
A_2\ot S^2A_3\op  (A_0\ot A_1\ot A_3)^{S}\\
&\op \La 3 A_0 \op A_0\ot \La 2A_2\op A_2\ot \La 2 A_1
\op A_2\ot \La 2 A_3  \op (A_0\ot A_1\ot A_3)^{\Lambda}, 
\end{align*}
for a total dimension of $43+20=63$.

$\Mthree$ has a nonzero projection onto each factor except the one-dimensional  $\La 3 A_0$.

When $n=4$ the summands are
\begin{align*}
&S^3A_0\op A_1\ot S^2A_2\op A_2\ot S^2A_4\op    A_3\ot S^2A_1 \op A_4\ot S^2A_3 \op (A_0\ot A_1\ot A_4)^S \op (A_0\ot A_2\ot A_3)^S
\\
&\op 
\La 3A_0\op A_1\ot \La 2A_3\op A_2\op \La 2A_4\ot A_3\ot \La 2A_1 \op A_4\ot \La 2A_3 \op (A_0\ot A_1\ot A_4)^{\Lambda} \op (A_0\ot A_2\ot A_3)^{\Lambda},
\end{align*}
for a total dimension of $164+112=276$, compared with the na\"\i ve search space dimension of $16^3=4096$ and the $\BZ_3$-invariant search 
space of dimension $1376$.

In summary:
\begin{proposition} The dimension of the space of $\BZ_{n+1}\times \BZ_3$ invariants in $(U^*\ot U)^{\ot 3}$ is
respectively of dimensions $10$, $63$, and $276$ when $n=2,3,4$. 
\end{proposition}

\section{Eigenvalues}\label{evalsect} When we deal with a restricted family   $PGL_n\cdot \cS$,  
it makes sense to discuss eigenlines  and eigenvalues of the terms appearing.
These also facilitate determining if two decompositions lie in the same family, beyond the graphs. 

\begin{table}
\begin{tabular}{|c|c|c|}
\hline
 \zfzt & Char.~Poly. & Count \\
 \hline
\multirow{3}{*}{symmetric}
 & $t^2(t - 1)$ & 6 \\
 & $t + t^2 + t^3$ & 4 \\
 & $t - t^2 + t^3 - 1$ & 1 \\
\hline\multirow{1}{*}{triples}
 & $\left\{ t^3,t^3,t^3 \right\}$ & 4 \\
\hline
\end{tabular}
\caption{Characteristic polynomials of matrices appearing in \zfzt}
\end{table}

\begin{table}
\begin{tabular}{|c|c|c|}
\hline
 \lad & Char.~Poly. & Count \\
 \hline
\multirow{2}{*}{symmetric}
 & $t^2(t - 1)$ & 4 \\
 & $t + t^2 + t^3$ & 1 \\
\hline\multirow{3}{*}{triples}
 & $\left\{ t^2(t - 1),t^3,t^3 \right\}$ & 3 \\
 & $\left\{ t^3,t^3,t + t^2 + t^3 \right\}$ & 1 \\
 & $\left\{ t^3,t^3,t - t^2 + t^3 - 1 \right\}$ & 2 \\
\hline
\end{tabular}
\caption{Characteristic polynomials of matrices appearing in \lad}
\end{table}

\begin{table}
\begin{tabular}{|c|c|c|}
\hline
 \twofix & Char.~Poly. & Count \\
 \hline
\multirow{2}{*}{symmetric}
 & $t(t - 1)^2$ & 1 \\
 & $t^2(t - 1)$ & 1 \\
\hline\multirow{5}{*}{triples}
 & $\left\{ t^3,t^3,t^3 \right\}$ & 1 \\
 & $\left\{ t^2(t - 1),t^2(t - 1),t + t^2 + t^3 \right\}$ & 1 \\
 & $\left\{ t^2(t - 1),t^3,t^3 \right\}$ & 2 \\
 & $\left\{ t^2(t - 1),t^2(t - 1),t^3 \right\}$ & 2 \\
 & $\left\{ t^2(t - 1),t^2(t - 1),t^2(t + 1) \right\}$ & 1 \\
\hline
\end{tabular}
\caption{Characteristic polynomials of matrices appearing in \twofix}
\end{table}

\bibliographystyle{amsplain}
 
\bibliography{Lmatrix}

\appendix

\section{Additional decompositions}\label{extraexs}
What follows are three additional decompositions with their graphs and eigenvalue tables.
 
\subsection{A decomposition with five $\BZ_3$-fixed points and no diagonal $GL_3$-symmetry}
The lack of extra symmetry may be easily deduced from the incidence graph.

\begin{tiny}
\begin{align}
\Mthree
&= \label{addtlone_c1} \begin{pmatrix} 0 & 0 & 0 \\ 0 & 0 & 0 \\ 0 & -1 & 1 \end{pmatrix}^{\ot 3} 
	+ \begin{pmatrix} 1 & 0 & 0 \\ 1 & 0 & 0 \\ 0 & 0 & 0 \end{pmatrix}^{\ot 3} 
	+ \begin{pmatrix} 1 & -1 & 0 \\ 0 & 0 & 0 \\ 0 & -1 & 0 \end{pmatrix}^{\ot 3} \\
&+ \label{addtlone_c4} \begin{pmatrix} 0 & 0 & 0 \\ 0 & 1 & 0 \\ 0 & 1 & 0 \end{pmatrix}^{\ot 3} 
	+ \begin{pmatrix} -1 & 1 & 0 \\ -1 & 0 & 0 \\ 0 & 1 & 0 \end{pmatrix}^{\ot 3} \\
&+ \label{addtlone_c6} \BZ_3^{std} \cdot 
	\begin{pmatrix} 0 & 0 & 0 \\ 0 & 0 & 1 \\ 0 & 0 & 0 \end{pmatrix} \ot 
	\begin{pmatrix} 0 & 0 & 0 \\ 0 & 0 & 0 \\ -1 & 0 & 0 \end{pmatrix} \ot 
	\begin{pmatrix} 0 & -1 & 1 \\ 0 & 0 & 0 \\ 0 & 0 & 0 \end{pmatrix} \\ 
&+ \label{addtlone_c7} \BZ_3^{std} \cdot 
	\begin{pmatrix} 0 & 0 & 0 \\ 0 & -1 & 1 \\ 0 & -1 & 1 \end{pmatrix} \ot 
	\begin{pmatrix} 1 & 0 & -1 \\ -1 & -1 & 1 \\ 0 & -1 & 0 \end{pmatrix} \ot 
	\begin{pmatrix} 0 & 0 & 0 \\ 0 & 0 & 0 \\ 0 & -1 & 0 \end{pmatrix} \\ 
&+ \label{addtlone_c8} \BZ_3^{std} \cdot 
	\begin{pmatrix} 0 & 0 & 0 \\ 1 & 1 & -1 \\ 1 & 1 & -1 \end{pmatrix} \ot 
	\begin{pmatrix} -1 & 0 & 1 \\ 0 & 0 & 0 \\ 0 & 0 & 0 \end{pmatrix} \ot 
	\begin{pmatrix} 0 & 0 & 0 \\ 0 & 0 & 1 \\ 0 & 0 & 1 \end{pmatrix} \\ 
&+ \label{addtlone_c9} \BZ_3^{std} \cdot 
	\begin{pmatrix} 0 & 0 & 0 \\ 0 & 0 & 1 \\ 0 & -1 & 1 \end{pmatrix} \ot 
	\begin{pmatrix} 0 & 0 & 0 \\ 1 & 1 & -1 \\ 0 & 1 & -1 \end{pmatrix} \ot 
	\begin{pmatrix} 1 & 0 & -1 \\ 0 & 0 & 0 \\ 0 & -1 & 0 \end{pmatrix} \\ 
&+ \label{addtlone_c10} \BZ_3^{std} \cdot 
	\begin{pmatrix} 0 & 0 & 0 \\ 0 & 0 & 0 \\ 1 & 0 & 0 \end{pmatrix} \ot 
	\begin{pmatrix} 1 & 0 & 0 \\ 0 & 0 & 0 \\ 0 & 0 & 0 \end{pmatrix} \ot 
	\begin{pmatrix} 0 & 0 & 1 \\ 0 & 0 & 1 \\ 0 & 0 & 1 \end{pmatrix} \\ 
&+ \label{addtlone_c11} \BZ_3^{std} \cdot 
	\begin{pmatrix} 0 & -1 & 0 \\ 0 & 0 & 0 \\ 0 & -1 & 0 \end{pmatrix} \ot 
	\begin{pmatrix} 1 & -1 & 0 \\ 1 & -1 & 0 \\ 0 & -1 & 0 \end{pmatrix} \ot 
	\begin{pmatrix} 0 & 0 & 0 \\ 1 & 0 & 0 \\ 0 & 0 & 0 \end{pmatrix} 
\end{align}
\end{tiny}

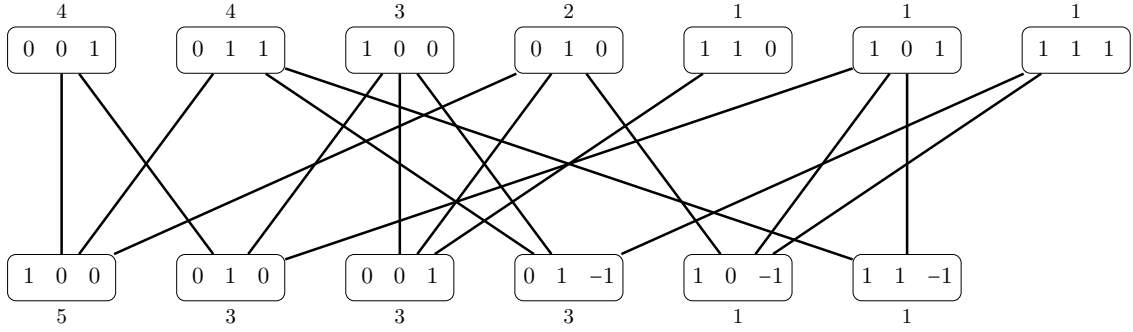
\begin{figure}
\begin{center}
\scalebox{.8}{\begin{tikzpicture}[node distance=3cm and 1cm, auto]

	\node[main node, label={\small 4}] (1u) {$\begin{matrix} 0 & 0 & 1 \end{matrix}$};
	\node[main node, label={\small 4}] (2u) [right=of 1u] {$\begin{matrix} 0 & 1 & 1 \end{matrix}$};
	\node[main node, label={\small 3}] (3u) [right=of 2u] {$\begin{matrix} 1 & 0 & 0 \end{matrix}$};
	\node[main node, label={\small 2}] (4u) [right=of 3u] {$\begin{matrix} 0 & 1 & 0 \end{matrix}$};
	\node[main node, label={\small 1}] (5u) [right=of 4u] {$\begin{matrix} 1 & 1 & 0 \end{matrix}$};
	\node[main node, label={\small 1}] (6u) [right=of 5u] {$\begin{matrix} 1 & 0 & 1 \end{matrix}$};
	\node[main node, label={\small 1}] (7u) [right=of 6u] {$\begin{matrix} 1 & 1 & 1 \end{matrix}$};

	\node[main node, label=below:{\small 5}] (1v) [below=of 1u] {$\begin{matrix} 1 & 0 & 0 \end{matrix}$};
	\node[main node, label=below:{\small 3}] (2v) [below=of 2u] {$\begin{matrix} 0 & 1 & 0 \end{matrix}$};
	\node[main node, label=below:{\small 3}] (3v) [below=of 3u] {$\begin{matrix} 0 & 0 & 1 \end{matrix}$};
	\node[main node, label=below:{\small 3}] (4v) [below=of 4u] {$\begin{matrix} 0 & 1 & -1 \end{matrix}$};
	\node[main node, label=below:{\small 1}] (5v) [below=of 5u] {$\begin{matrix} 1 & 0 & -1 \end{matrix}$};
	\node[main node, label=below:{\small 1}] (6v) [below=of 6u] {$\begin{matrix} 1 & 1 & -1 \end{matrix}$};

	\path[line]
	(1u)
		edge node {} (1v)
		edge node {} (2v)
	(2u)
		edge node {} (1v)
		edge node {} (4v)
		edge node {} (6v)
	(3u)
		edge node {} (2v)
		edge node {} (3v)
		edge node {} (4v)
	(4u)
		edge node {} (1v)
		edge node {} (3v)
		edge node {} (5v)
	(5u)
		edge node {} (3v)
	(6u)
		edge node {} (2v)
		edge node {} (5v)
		edge node {} (6v)
	(7u)
		edge node {} (4v)
		edge node {} (5v)
	;

\end{tikzpicture}}
\end{center}
\caption{Incidence graph of \addtlone}
\label{fig:addtlone_inc}
\end{figure}

\begin{figure}
\begin{center}
\scalebox{.8}{\begin{tikzpicture}[node distance=3cm and 1cm, auto]

	\node[main node, label={\small 4}] (1u) {$\begin{matrix} 0 & 0 & 1 \end{matrix}$};
	\node[main node, label={\small 4}] (2u) [right=of 1u] {$\begin{matrix} 0 & 1 & 1 \end{matrix}$};
	\node[main node, label={\small 3}] (3u) [right=of 2u] {$\begin{matrix} 1 & 0 & 0 \end{matrix}$};
	\node[main node, label={\small 2}] (4u) [right=of 3u] {$\begin{matrix} 0 & 1 & 0 \end{matrix}$};
	\node[main node, label={\small 1}] (5u) [right=of 4u] {$\begin{matrix} 1 & 1 & 0 \end{matrix}$};
	\node[main node, label={\small 1}] (6u) [right=of 5u] {$\begin{matrix} 1 & 0 & 1 \end{matrix}$};
	\node[main node, label={\small 1}] (7u) [right=of 6u] {$\begin{matrix} 1 & 1 & 1 \end{matrix}$};

	\node[main node, label=below:{\small 5}] (1v) [below=of 1u] {$\begin{matrix} 1 & 0 & 0 \end{matrix}$};
	\node[main node, label=below:{\small 3}] (2v) [below=of 2u] {$\begin{matrix} 0 & 1 & 0 \end{matrix}$};
	\node[main node, label=below:{\small 3}] (3v) [below=of 3u] {$\begin{matrix} 0 & 0 & 1 \end{matrix}$};
	\node[main node, label=below:{\small 3}] (4v) [below=of 4u] {$\begin{matrix} 0 & 1 & -1 \end{matrix}$};
	\node[main node, label=below:{\small 1}] (5v) [below=of 5u] {$\begin{matrix} 1 & 0 & -1 \end{matrix}$};
	\node[main node, label=below:{\small 1}] (6v) [below=of 6u] {$\begin{matrix} 1 & 1 & -1 \end{matrix}$};

	\path[line]
	(1u) edge [dashed,black,] node {} (4v)
	(5u) edge [dashed,black,] node {} (1v)
	(2u) edge [dashed,black,] node {} (2v)
	(4u) edge [blue,] node {} (3v)
	(1u) edge [blue,] node {} (1v)
	(3u) edge [blue,] node {} (4v)
	(2u) edge [red,] node {} (4v)
	(1u) edge [red,] node {} (2v)
	(2u) edge [green,] node {} (6v)
	(3u) edge [green,] node {} (5v)
	(2u) edge [green,] node {} (3v)
	(1u) edge [purple,bend left=8] node {} (1v)
	(3u) edge [purple,] node {} (1v)
	(7u) edge [purple,] node {} (3v)
	(6u) edge [orange,] node {} (2v)
	(4u) edge [orange,] node {} (1v)
	;

\end{tikzpicture}}
\end{center}
\caption{Pairing graph of \addtlone}
\label{fig:addtlone_pair}
\end{figure}
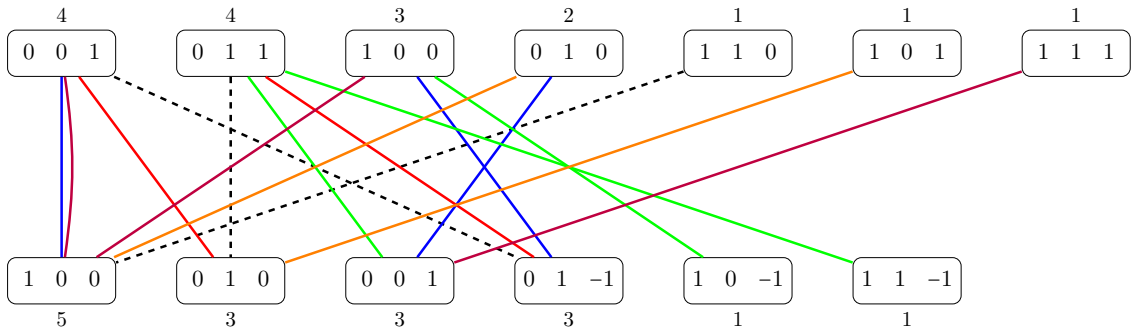

\begin{table}
\begin{tabular}{|c|c|c|}
\hline
 \addtlone & Char.~Poly. & Count \\
 \hline
\multirow{2}{*}{symmetric}
 & $t^2(t - 1)$ & 4 \\
 & $t + t^2 + t^3$ & 1 \\
\hline\multirow{4}{*}{triples}
 & $\left\{ t^3,t^3,t^3 \right\}$ & 3 \\
 & $\left\{ t^2(t - 1),t^2(t + 1),t^3 \right\}$ & 1 \\
 & $\left\{ t^2(t - 1),t^3,t - t^2 + t^3 \right\}$ & 1 \\
 & $\left\{ t^2(t - 1),t^2(t - 1),t^3 \right\}$ & 1 \\
\hline
\end{tabular}
\caption{Characteristic polynomials of matrices appearing in \addtlone}
\end{table}

\subsection{Another decomposition with five $\BZ_3$-fixed points}\ 

\begin{tiny}
\begin{align}
\Mthree
&= \label{addtltwo_c1} \begin{pmatrix} 0 & 1 & -1 \\ 0 & 1 & 0 \\ 0 & 0 & 0 \end{pmatrix}^{\ot 3} 
	+ \begin{pmatrix} 0 & 0 & 0 \\ -1 & 1 & 1 \\ 0 & 0 & 0 \end{pmatrix}^{\ot 3} 
	+ \begin{pmatrix} 1 & 0 & 0 \\ 0 & 0 & 0 \\ 0 & 0 & 0 \end{pmatrix}^{\ot 3} \\
&+ \label{addtltwo_c4} \begin{pmatrix} 0 & -1 & 1 \\ 1 & -1 & -1 \\ 0 & 0 & 0 \end{pmatrix}^{\ot 3} 
	+ \begin{pmatrix} 0 & 0 & 0 \\ 0 & 0 & 0 \\ 0 & 0 & 1 \end{pmatrix}^{\ot 3} \\
&+ \label{addtltwo_c6} \BZ_3^{std} \cdot 
	\begin{pmatrix} -1 & 0 & 1 \\ 0 & 0 & 0 \\ -1 & 0 & 1 \end{pmatrix} \ot 
	\begin{pmatrix} 0 & 0 & 0 \\ 0 & 1 & 0 \\ 0 & 1 & 0 \end{pmatrix} \ot 
	\begin{pmatrix} 0 & 0 & 0 \\ 0 & 0 & 1 \\ 0 & 0 & 0 \end{pmatrix} \\ 
&+ \label{addtltwo_c7} \BZ_3^{std} \cdot 
	\begin{pmatrix} 0 & -1 & 1 \\ 0 & 0 & 0 \\ 0 & 0 & 0 \end{pmatrix} \ot 
	\begin{pmatrix} 0 & 0 & 0 \\ -1 & 0 & 1 \\ 0 & 0 & 0 \end{pmatrix} \ot 
	\begin{pmatrix} 1 & -1 & 0 \\ 1 & -1 & -1 \\ 0 & 0 & 0 \end{pmatrix} \\ 
&+ \label{addtltwo_c8} \BZ_3^{std} \cdot 
	\begin{pmatrix} 0 & 0 & 0 \\ 1 & 0 & 0 \\ 1 & 0 & 0 \end{pmatrix} \ot 
	\begin{pmatrix} -1 & 0 & 0 \\ 0 & 0 & 0 \\ -1 & 0 & 0 \end{pmatrix} \ot 
	\begin{pmatrix} 0 & -1 & 0 \\ 0 & -1 & 0 \\ 0 & -1 & 0 \end{pmatrix} \\ 
&+ \label{addtltwo_c9} \BZ_3^{std} \cdot 
	\begin{pmatrix} 0 & 0 & 0 \\ 0 & 0 & -1 \\ 0 & 0 & 0 \end{pmatrix} \ot 
	\begin{pmatrix} 0 & 1 & -1 \\ 0 & 0 & 0 \\ 0 & 1 & -1 \end{pmatrix} \ot 
	\begin{pmatrix} 0 & 0 & 0 \\ 0 & -1 & 0 \\ 0 & 0 & 0 \end{pmatrix} \\ 
&+ \label{addtltwo_c10} \BZ_3^{std} \cdot 
	\begin{pmatrix} 1 & 0 & -1 \\ 0 & 0 & 0 \\ 1 & 0 & 0 \end{pmatrix} \ot 
	\begin{pmatrix} 0 & 1 & -1 \\ 0 & 1 & 0 \\ 0 & 1 & 0 \end{pmatrix} \ot 
	\begin{pmatrix} 0 & 0 & 0 \\ -1 & 0 & 1 \\ -1 & 0 & 0 \end{pmatrix} \\ 
&+ \label{addtltwo_c11} \BZ_3^{std} \cdot 
	\begin{pmatrix} 0 & -1 & 1 \\ -1 & -1 & 1 \\ -1 & -1 & 1 \end{pmatrix} \ot 
	\begin{pmatrix} 0 & 0 & 0 \\ 0 & 0 & 0 \\ -1 & 0 & 0 \end{pmatrix} \ot 
	\begin{pmatrix} 0 & 0 & -1 \\ 0 & 0 & 0 \\ 0 & 0 & 0 \end{pmatrix} 
\end{align}
\end{tiny}

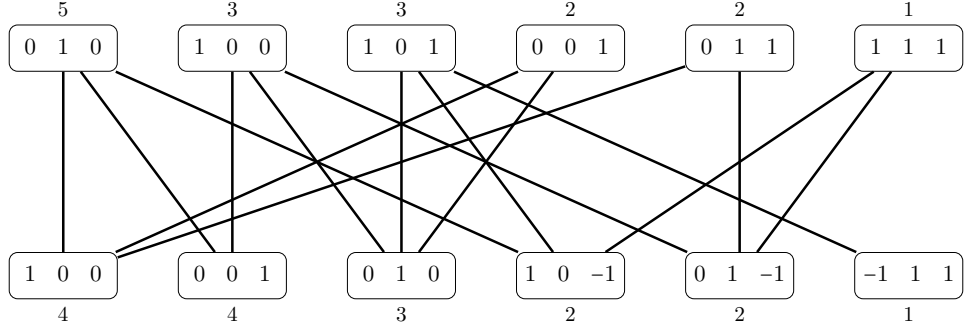
\begin{figure}
\begin{center}
\scalebox{.8}{\begin{tikzpicture}[node distance=3cm and 1cm, auto]

	\node[main node, label={\small 5}] (1u) {$\begin{matrix} 0 & 1 & 0 \end{matrix}$};
	\node[main node, label={\small 3}] (2u) [right=of 1u] {$\begin{matrix} 1 & 0 & 0 \end{matrix}$};
	\node[main node, label={\small 3}] (3u) [right=of 2u] {$\begin{matrix} 1 & 0 & 1 \end{matrix}$};
	\node[main node, label={\small 2}] (4u) [right=of 3u] {$\begin{matrix} 0 & 0 & 1 \end{matrix}$};
	\node[main node, label={\small 2}] (5u) [right=of 4u] {$\begin{matrix} 0 & 1 & 1 \end{matrix}$};
	\node[main node, label={\small 1}] (6u) [right=of 5u] {$\begin{matrix} 1 & 1 & 1 \end{matrix}$};

	\node[main node, label=below:{\small 4}] (1v) [below=of 1u] {$\begin{matrix} 1 & 0 & 0 \end{matrix}$};
	\node[main node, label=below:{\small 4}] (2v) [below=of 2u] {$\begin{matrix} 0 & 0 & 1 \end{matrix}$};
	\node[main node, label=below:{\small 3}] (3v) [below=of 3u] {$\begin{matrix} 0 & 1 & 0 \end{matrix}$};
	\node[main node, label=below:{\small 2}] (4v) [below=of 4u] {$\begin{matrix} 1 & 0 & -1 \end{matrix}$};
	\node[main node, label=below:{\small 2}] (5v) [below=of 5u] {$\begin{matrix} 0 & 1 & -1 \end{matrix}$};
	\node[main node, label=below:{\small 1}] (6v) [below=of 6u] {$\begin{matrix} -1 & 1 & 1 \end{matrix}$};

	\path[line]
	(1u)
		edge node {} (1v)
		edge node {} (2v)
		edge node {} (4v)
	(2u)
		edge node {} (2v)
		edge node {} (3v)
		edge node {} (5v)
	(3u)
		edge node {} (3v)
		edge node {} (4v)
		edge node {} (6v)
	(4u)
		edge node {} (1v)
		edge node {} (3v)
	(5u)
		edge node {} (1v)
		edge node {} (5v)
	(6u)
		edge node {} (4v)
		edge node {} (5v)
	;

\end{tikzpicture}}
\end{center}
\caption{Incidence graph of \addtltwo}
\label{fig:addtltwo_inc}
\end{figure}

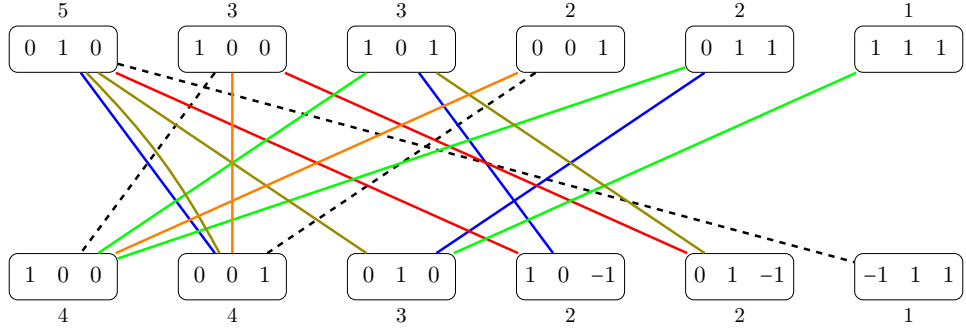
\begin{figure}
\begin{center}
\scalebox{.8}{\begin{tikzpicture}[node distance=3cm and 1cm, auto]

	\node[main node, label={\small 5}] (1u) {$\begin{matrix} 0 & 1 & 0 \end{matrix}$};
	\node[main node, label={\small 3}] (2u) [right=of 1u] {$\begin{matrix} 1 & 0 & 0 \end{matrix}$};
	\node[main node, label={\small 3}] (3u) [right=of 2u] {$\begin{matrix} 1 & 0 & 1 \end{matrix}$};
	\node[main node, label={\small 2}] (4u) [right=of 3u] {$\begin{matrix} 0 & 0 & 1 \end{matrix}$};
	\node[main node, label={\small 2}] (5u) [right=of 4u] {$\begin{matrix} 0 & 1 & 1 \end{matrix}$};
	\node[main node, label={\small 1}] (6u) [right=of 5u] {$\begin{matrix} 1 & 1 & 1 \end{matrix}$};

	\node[main node, label=below:{\small 4}] (1v) [below=of 1u] {$\begin{matrix} 1 & 0 & 0 \end{matrix}$};
	\node[main node, label=below:{\small 4}] (2v) [below=of 2u] {$\begin{matrix} 0 & 0 & 1 \end{matrix}$};
	\node[main node, label=below:{\small 3}] (3v) [below=of 3u] {$\begin{matrix} 0 & 1 & 0 \end{matrix}$};
	\node[main node, label=below:{\small 2}] (4v) [below=of 4u] {$\begin{matrix} 1 & 0 & -1 \end{matrix}$};
	\node[main node, label=below:{\small 2}] (5v) [below=of 5u] {$\begin{matrix} 0 & 1 & -1 \end{matrix}$};
	\node[main node, label=below:{\small 1}] (6v) [below=of 6u] {$\begin{matrix} -1 & 1 & 1 \end{matrix}$};

	\path[line]
	(1u) edge [dashed,black,] node {} (6v)
	(2u) edge [dashed,black,] node {} (1v)
	(4u) edge [dashed,black,] node {} (2v)
	(3u) edge [blue,] node {} (4v)
	(5u) edge [blue,] node {} (3v)
	(1u) edge [blue,] node {} (2v)
	(2u) edge [red,] node {} (5v)
	(1u) edge [red,] node {} (4v)
	(5u) edge [green,] node {} (1v)
	(3u) edge [green,] node {} (1v)
	(6u) edge [green,] node {} (3v)
	(1u) edge [olive,bend left=8] node {} (2v)
	(3u) edge [olive,] node {} (5v)
	(1u) edge [olive,] node {} (3v)
	(4u) edge [orange,] node {} (1v)
	(2u) edge [orange,] node {} (2v)
	;

\end{tikzpicture}}
\end{center}
\caption{Pairing graph of \addtltwo}
\label{fig:addtltwo_pair}
\end{figure}

\begin{table}
\begin{tabular}{|c|c|c|}
\hline
 \addtltwo & Char.~Poly. & Count \\
 \hline
\multirow{2}{*}{symmetric}
 & $t^2(t - 1)$ & 4 \\
 & $t + t^2 + t^3$ & 1 \\
\hline\multirow{4}{*}{triples}
 & $\left\{ t^2(t - 1),t^3,t^3 \right\}$ & 1 \\
 & $\left\{ t^3,t^3,t^3 \right\}$ & 2 \\
 & $\left\{ t^2(t + 1),t^2(t + 1),t^3 \right\}$ & 2 \\
 & $\left\{ t^2(t - 1),t^3,t - t^2 + t^3 \right\}$ & 1 \\
\hline
\end{tabular}
\caption{Characteristic polynomials of matrices appearing in \addtltwo}
\end{table}

\subsection{Another decomposition with $2$ $\BZ_3$-fixed points} \ 

\begin{tiny}
\begin{align}
\Mthree
&= \label{addtlthree_c1} \begin{pmatrix} 0 & 0 & 0 \\ 0 & 1 & 0 \\ 0 & 0 & 1 \end{pmatrix}^{\ot 3} 
	+ \begin{pmatrix} 1 & 0 & 0 \\ 0 & 0 & 0 \\ 0 & 0 & 0 \end{pmatrix}^{\ot 3} \\
&+ \label{addtlthree_c3} \BZ_3^{std} \cdot 
	\begin{pmatrix} 0 & 0 & 0 \\ 0 & 0 & 0 \\ 0 & 0 & 1 \end{pmatrix} \ot 
	\begin{pmatrix} 0 & 0 & 0 \\ 0 & -1 & 0 \\ 0 & 1 & 0 \end{pmatrix} \ot 
	\begin{pmatrix} 0 & 0 & 0 \\ 0 & 1 & 1 \\ 0 & 0 & 0 \end{pmatrix} \\ 
&+ \label{addtlthree_c4} \BZ_3^{std} \cdot 
	\begin{pmatrix} 0 & 0 & 0 \\ 1 & 1 & 0 \\ 0 & 0 & 0 \end{pmatrix} \ot 
	\begin{pmatrix} -1 & 0 & -1 \\ 0 & 0 & 0 \\ 0 & 0 & 0 \end{pmatrix} \ot 
	\begin{pmatrix} 0 & -1 & 0 \\ 0 & 0 & 0 \\ 0 & 0 & 0 \end{pmatrix} \\ 
&+ \label{addtlthree_c5} \BZ_3^{std} \cdot 
	\begin{pmatrix} 0 & 0 & 0 \\ 0 & 0 & 0 \\ -1 & 0 & 0 \end{pmatrix} \ot 
	\begin{pmatrix} 0 & 1 & 0 \\ 0 & 0 & 0 \\ 0 & 0 & 0 \end{pmatrix} \ot 
	\begin{pmatrix} 0 & 0 & 0 \\ 0 & -1 & -1 \\ 0 & 0 & 0 \end{pmatrix} \\ 
&+ \label{addtlthree_c6} \BZ_3^{std} \cdot 
	\begin{pmatrix} 0 & 0 & -1 \\ 0 & 0 & 0 \\ 0 & 0 & 0 \end{pmatrix} \ot 
	\begin{pmatrix} 0 & 1 & 0 \\ 0 & 0 & 0 \\ -1 & -1 & 0 \end{pmatrix} \ot 
	\begin{pmatrix} 0 & 0 & 0 \\ 1 & 0 & 0 \\ -1 & 0 & 0 \end{pmatrix} \\ 
&+ \label{addtlthree_c7} \BZ_3^{std} \cdot 
	\begin{pmatrix} 0 & 1 & 1 \\ 0 & 0 & 0 \\ 0 & -1 & -1 \end{pmatrix} \ot 
	\begin{pmatrix} 0 & 0 & 0 \\ 0 & 0 & 0 \\ 1 & 0 & 0 \end{pmatrix} \ot 
	\begin{pmatrix} 1 & 0 & 0 \\ -1 & 0 & 0 \\ 1 & 0 & 0 \end{pmatrix} \\ 
&+ \label{addtlthree_c8} \BZ_3^{std} \cdot 
	\begin{pmatrix} 0 & 0 & 0 \\ 0 & 0 & 0 \\ 0 & 1 & 1 \end{pmatrix} \ot 
	\begin{pmatrix} 0 & 0 & 0 \\ 0 & 1 & 0 \\ 0 & 0 & 0 \end{pmatrix} \ot 
	\begin{pmatrix} -1 & 0 & -1 \\ 1 & 0 & 1 \\ -1 & 0 & -1 \end{pmatrix} \\ 
&+ \label{addtlthree_c9} \BZ_3^{std} \cdot 
	\begin{pmatrix} 0 & 0 & 0 \\ 0 & -1 & 0 \\ -1 & 0 & 0 \end{pmatrix} \ot 
	\begin{pmatrix} 1 & 0 & 1 \\ -1 & 0 & 0 \\ 1 & 0 & 0 \end{pmatrix} \ot 
	\begin{pmatrix} 0 & 1 & 0 \\ 0 & 0 & 0 \\ 0 & -1 & -1 \end{pmatrix} 
\end{align}
\end{tiny}

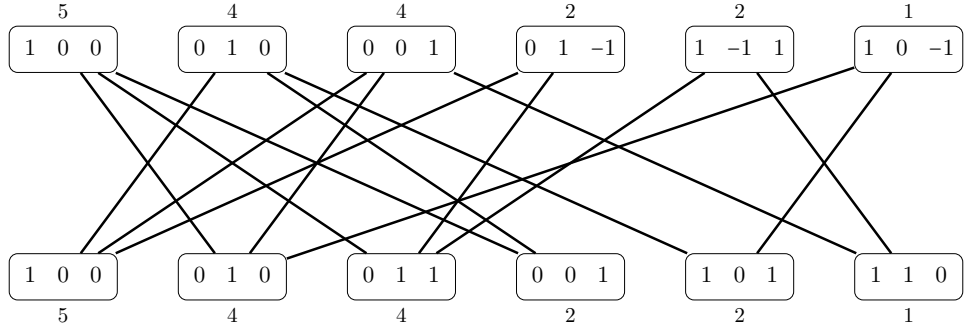
\begin{figure}
\begin{center}
\scalebox{.8}{\begin{tikzpicture}[node distance=3cm and 1cm, auto]

	\node[main node, label={\small 5}] (1u) {$\begin{matrix} 1 & 0 & 0 \end{matrix}$};
	\node[main node, label={\small 4}] (2u) [right=of 1u] {$\begin{matrix} 0 & 1 & 0 \end{matrix}$};
	\node[main node, label={\small 4}] (3u) [right=of 2u] {$\begin{matrix} 0 & 0 & 1 \end{matrix}$};
	\node[main node, label={\small 2}] (4u) [right=of 3u] {$\begin{matrix} 0 & 1 & -1 \end{matrix}$};
	\node[main node, label={\small 2}] (5u) [right=of 4u] {$\begin{matrix} 1 & -1 & 1 \end{matrix}$};
	\node[main node, label={\small 1}] (6u) [right=of 5u] {$\begin{matrix} 1 & 0 & -1 \end{matrix}$};

	\node[main node, label=below:{\small 5}] (1v) [below=of 1u] {$\begin{matrix} 1 & 0 & 0 \end{matrix}$};
	\node[main node, label=below:{\small 4}] (2v) [below=of 2u] {$\begin{matrix} 0 & 1 & 0 \end{matrix}$};
	\node[main node, label=below:{\small 4}] (3v) [below=of 3u] {$\begin{matrix} 0 & 1 & 1 \end{matrix}$};
	\node[main node, label=below:{\small 2}] (4v) [below=of 4u] {$\begin{matrix} 0 & 0 & 1 \end{matrix}$};
	\node[main node, label=below:{\small 2}] (5v) [below=of 5u] {$\begin{matrix} 1 & 0 & 1 \end{matrix}$};
	\node[main node, label=below:{\small 1}] (6v) [below=of 6u] {$\begin{matrix} 1 & 1 & 0 \end{matrix}$};

	\path[line]
	(1u)
		edge node {} (2v)
		edge node {} (3v)
		edge node {} (4v)
	(2u)
		edge node {} (1v)
		edge node {} (4v)
		edge node {} (5v)
	(3u)
		edge node {} (1v)
		edge node {} (2v)
		edge node {} (6v)
	(4u)
		edge node {} (1v)
		edge node {} (3v)
	(5u)
		edge node {} (3v)
		edge node {} (6v)
	(6u)
		edge node {} (2v)
		edge node {} (5v)
	;

\end{tikzpicture}}
\end{center}
\caption{Incidence graph of \addtlthree}
\label{fig:addtlthree_inc}
\end{figure}

\begin{figure}
\begin{center}
\scalebox{.8}{\begin{tikzpicture}[node distance=3cm and 1cm, auto]

	\node[main node, label={\small 5}] (1u) {$\begin{matrix} 1 & 0 & 0 \end{matrix}$};
	\node[main node, label={\small 4}] (2u) [right=of 1u] {$\begin{matrix} 0 & 1 & 0 \end{matrix}$};
	\node[main node, label={\small 4}] (3u) [right=of 2u] {$\begin{matrix} 0 & 0 & 1 \end{matrix}$};
	\node[main node, label={\small 2}] (4u) [right=of 3u] {$\begin{matrix} 0 & 1 & -1 \end{matrix}$};
	\node[main node, label={\small 2}] (5u) [right=of 4u] {$\begin{matrix} 1 & -1 & 1 \end{matrix}$};
	\node[main node, label={\small 1}] (6u) [right=of 5u] {$\begin{matrix} 1 & 0 & -1 \end{matrix}$};

	\node[main node, label=below:{\small 5}] (1v) [below=of 1u] {$\begin{matrix} 1 & 0 & 0 \end{matrix}$};
	\node[main node, label=below:{\small 4}] (2v) [below=of 2u] {$\begin{matrix} 0 & 1 & 0 \end{matrix}$};
	\node[main node, label=below:{\small 4}] (3v) [below=of 3u] {$\begin{matrix} 0 & 1 & 1 \end{matrix}$};
	\node[main node, label=below:{\small 2}] (4v) [below=of 4u] {$\begin{matrix} 0 & 0 & 1 \end{matrix}$};
	\node[main node, label=below:{\small 2}] (5v) [below=of 5u] {$\begin{matrix} 1 & 0 & 1 \end{matrix}$};
	\node[main node, label=below:{\small 1}] (6v) [below=of 6u] {$\begin{matrix} 1 & 1 & 0 \end{matrix}$};

	\path[line]
	(1u) edge [dashed,black,] node {} (1v)
	(3u) edge [blue,] node {} (4v)
	(4u) edge [blue,] node {} (2v)
	(2u) edge [blue,] node {} (3v)
	(2u) edge [red,] node {} (6v)
	(1u) edge [red,] node {} (5v)
	(1u) edge [red,] node {} (2v)
	(3u) edge [green,] node {} (1v)
	(1u) edge [green,bend left=8] node {} (2v)
	(2u) edge [green,bend left=8] node {} (3v)
	(1u) edge [olive,] node {} (4v)
	(4u) edge [olive,] node {} (1v)
	(6u) edge [purple,] node {} (3v)
	(3u) edge [purple,bend left=8] node {} (1v)
	(5u) edge [purple,] node {} (1v)
	(3u) edge [orange,] node {} (3v)
	(2u) edge [orange,] node {} (2v)
	(5u) edge [orange,] node {} (5v)
	;

\end{tikzpicture}}
\end{center}
\caption{Pairing graph of \addtlthree}
\label{fig:addtlthree_pair}
\end{figure}
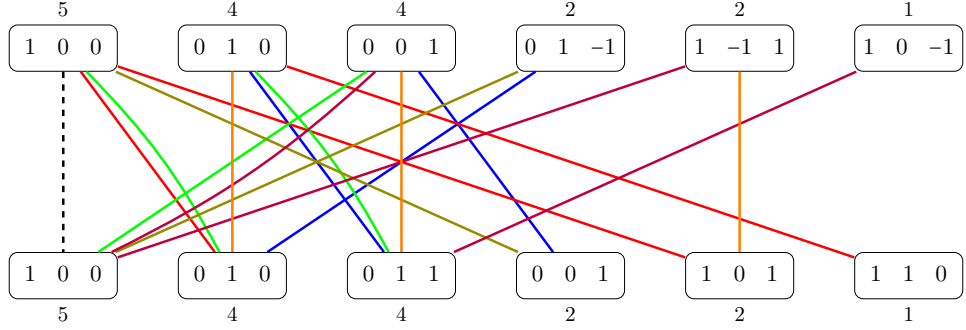

\begin{table}
\begin{tabular}{|c|c|c|}
\hline
 \addtlthree & Char.~Poly. & Count \\
 \hline
\multirow{2}{*}{symmetric}
 & $t(t - 1)^2$ & 1 \\
 & $t^2(t - 1)$ & 1 \\
\hline\multirow{6}{*}{triples}
 & $\left\{ t^2(t - 1),t^2(t - 1),t^2(t + 1) \right\}$ & 1 \\
 & $\left\{ t^2(t - 1),t^2(t + 1),t^3 \right\}$ & 2 \\
 & $\left\{ t^2(t + 1),t^3,t^3 \right\}$ & 1 \\
 & $\left\{ t^3,t^3,t^3 \right\}$ & 1 \\
 & $\left\{ t^2(t - 1),t^2(t - 1),2t^2 + t^3 \right\}$ & 1 \\
 & $\left\{ t^2(t + 1),t^2(t + 1),t^3 - t^2 - t \right\}$ & 1 \\
\hline
\end{tabular}
\caption{Characteristic polynomials of matrices appearing in \addtlthree}
\end{table}

\section{Triplets}

For the reader's convenience, we write out all the matrices appearing $\cS_{\BZ_4\times \BZ_3}$ and $\cS_{Lader-\BZ_3^{std}}$.

Matrix triplets for \zfzt:
\begin{tiny}
\begin{align}
\Mthree
&= \label{zfzt_c1} \begin{pmatrix} 0 & 0 & 1 \\ -1 & 0 & 1 \\ 0 & -1 & 1 \end{pmatrix}^{\ot 3} 
	+ \begin{pmatrix} 0 & 1 & 0 \\ 0 & 1 & 0 \\ 0 & 0 & 0 \end{pmatrix}^{\ot 3} 
	+ \begin{pmatrix} 0 & 0 & 0 \\ -1 & 0 & 1 \\ -1 & 0 & 1 \end{pmatrix}^{\ot 3} \\
&+ \label{zfzt_c4} \begin{pmatrix} 1 & 0 & 0 \\ 0 & 0 & 0 \\ 0 & 0 & 0 \end{pmatrix}^{\ot 3} 
	+ \begin{pmatrix} 0 & 0 & 0 \\ -1 & 1 & 0 \\ 0 & 0 & 0 \end{pmatrix}^{\ot 3} 
	+ \begin{pmatrix} 0 & 0 & 0 \\ 0 & 0 & 0 \\ 0 & -1 & 1 \end{pmatrix}^{\ot 3} \\
&+ \label{zfzt_c7} \begin{pmatrix} 0 & 0 & 1 \\ 0 & 0 & 1 \\ 0 & 0 & 1 \end{pmatrix}^{\ot 3} 
	+ \begin{pmatrix} 0 & -1 & 0 \\ 1 & -1 & 0 \\ 0 & 0 & 0 \end{pmatrix}^{\ot 3} 
	+ \begin{pmatrix} 0 & 0 & 0 \\ 1 & 0 & -1 \\ 0 & 1 & -1 \end{pmatrix}^{\ot 3} \\
&+ \label{zfzt_c10} \begin{pmatrix} 0 & 0 & -1 \\ 0 & 0 & -1 \\ 0 & 1 & -1 \end{pmatrix}^{\ot 3} 
	+ \begin{pmatrix} 0 & 0 & -1 \\ 1 & 0 & -1 \\ 1 & 0 & -1 \end{pmatrix}^{\ot 3} \\
&+ \label{zfzt_c12} \BZ_3^{std} \cdot 
	\begin{pmatrix} 0 & 0 & 0 \\ 0 & 0 & 1 \\ 0 & 0 & 0 \end{pmatrix} \ot 
	\begin{pmatrix} 0 & 0 & 0 \\ 0 & 0 & 0 \\ -1 & 1 & 0 \end{pmatrix} \ot 
	\begin{pmatrix} 0 & 0 & 0 \\ 0 & 1 & -1 \\ 0 & 1 & -1 \end{pmatrix} \\ 
&+ \label{zfzt_c13} \BZ_3^{std} \cdot 
	\begin{pmatrix} 0 & 0 & 0 \\ 0 & 0 & 0 \\ -1 & 0 & 0 \end{pmatrix} \ot 
	\begin{pmatrix} 0 & 1 & -1 \\ 0 & 1 & -1 \\ 0 & 1 & -1 \end{pmatrix} \ot 
	\begin{pmatrix} 0 & 0 & -1 \\ 0 & 0 & -1 \\ 0 & 0 & 0 \end{pmatrix} \\ 
&+ \label{zfzt_c14} \BZ_3^{std} \cdot 
	\begin{pmatrix} -1 & 1 & 0 \\ -1 & 1 & 0 \\ -1 & 1 & 0 \end{pmatrix} \ot 
	\begin{pmatrix} 0 & 0 & -1 \\ 0 & 0 & 0 \\ 0 & 0 & 0 \end{pmatrix} \ot 
	\begin{pmatrix} 0 & 0 & 0 \\ 1 & 0 & 0 \\ 1 & 0 & 0 \end{pmatrix} \\ 
&+ \label{zfzt_c15} \BZ_3^{std} \cdot 
	\begin{pmatrix} 0 & 1 & -1 \\ 0 & 0 & 0 \\ 0 & 0 & 0 \end{pmatrix} \ot 
	\begin{pmatrix} 0 & 0 & 0 \\ 1 & 0 & 0 \\ 0 & 0 & 0 \end{pmatrix} \ot 
	\begin{pmatrix} 1 & -1 & 0 \\ 1 & -1 & 0 \\ 0 & 0 & 0 \end{pmatrix} 
\end{align}

\end{tiny}

Matrix triplets for \lad:
\begin{tiny}
\begin{align}
\Mthree
&= \label{lad_c1} \begin{pmatrix} 0 & 0 & 0 \\ 0 & 1 & 0 \\ 0 & 0 & 0 \end{pmatrix}^{\ot 3} 
	+ \begin{pmatrix} 0 & 0 & 0 \\ 0 & 0 & 0 \\ 0 & 0 & 1 \end{pmatrix}^{\ot 3} 
	+ \begin{pmatrix} -1 & 1 & 0 \\ -1 & 0 & 0 \\ 0 & 0 & 0 \end{pmatrix}^{\ot 3} \\
&+ \label{lad_c4} \begin{pmatrix} 1 & 0 & 0 \\ 1 & 0 & 0 \\ 0 & 0 & 0 \end{pmatrix}^{\ot 3} 
	+ \begin{pmatrix} 1 & -1 & 0 \\ 0 & 0 & 0 \\ 0 & 0 & 0 \end{pmatrix}^{\ot 3} \\
&+ \label{lad_c6} \BZ_3^{std} \cdot 
	\begin{pmatrix} 1 & 0 & 0 \\ 0 & 0 & 0 \\ 0 & 0 & 0 \end{pmatrix} \ot 
	\begin{pmatrix} 0 & 0 & 1 \\ 0 & 0 & 0 \\ 0 & 0 & 0 \end{pmatrix} \ot 
	\begin{pmatrix} 0 & 0 & 0 \\ 0 & 0 & 0 \\ 1 & 0 & 0 \end{pmatrix} \\ 
&+ \label{lad_c7} \BZ_3^{std} \cdot 
	\begin{pmatrix} 0 & 0 & 0 \\ 0 & 0 & 1 \\ 0 & -1 & -1 \end{pmatrix} \ot 
	\begin{pmatrix} 0 & 0 & 0 \\ 1 & 0 & 0 \\ -1 & 1 & 0 \end{pmatrix} \ot 
	\begin{pmatrix} 0 & -1 & -1 \\ 0 & 0 & -1 \\ 0 & 0 & 0 \end{pmatrix} \\ 
&+ \label{lad_c8} \BZ_3^{std} \cdot 
	\begin{pmatrix} 0 & -1 & 0 \\ 0 & 0 & 0 \\ 0 & 0 & 0 \end{pmatrix} \ot 
	\begin{pmatrix} 1 & -1 & 0 \\ 1 & -1 & -1 \\ 0 & 1 & 1 \end{pmatrix} \ot 
	\begin{pmatrix} 0 & 0 & 0 \\ 1 & 0 & 0 \\ 0 & 0 & 0 \end{pmatrix} \\ 
&+ \label{lad_c9} \BZ_3^{std} \cdot 
	\begin{pmatrix} 0 & 0 & 0 \\ 0 & 0 & 0 \\ 0 & -1 & 0 \end{pmatrix} \ot 
	\begin{pmatrix} 0 & 1 & 1 \\ -1 & 1 & 1 \\ 1 & -1 & 0 \end{pmatrix} \ot 
	\begin{pmatrix} 0 & 0 & 0 \\ 0 & 0 & -1 \\ 0 & 0 & 0 \end{pmatrix} \\ 
&+ \label{lad_c10} \BZ_3^{std} \cdot 
	\begin{pmatrix} 0 & 0 & 0 \\ 0 & 0 & -1 \\ 0 & 0 & 1 \end{pmatrix} \ot 
	\begin{pmatrix} 0 & 0 & 0 \\ -1 & 0 & 0 \\ 1 & 0 & 0 \end{pmatrix} \ot 
	\begin{pmatrix} 0 & 0 & 1 \\ 0 & 0 & 1 \\ 0 & 0 & 0 \end{pmatrix} \\ 
&+ \label{lad_c11} \BZ_3^{std} \cdot 
	\begin{pmatrix} 0 & 0 & 0 \\ 0 & 0 & 0 \\ 0 & 1 & 1 \end{pmatrix} \ot 
	\begin{pmatrix} 0 & 0 & 0 \\ 0 & 0 & 0 \\ 1 & -1 & 0 \end{pmatrix} \ot 
	\begin{pmatrix} 0 & 1 & 1 \\ 0 & 0 & 0 \\ 0 & 0 & 0 \end{pmatrix} 
\end{align}

\end{tiny}

\end{document}